\documentclass[a4paper,UKenglish,cleveref, autoref]{lipics-v2019}
\usepackage{amsmath}
\usepackage{stmaryrd} %
\usepackage{amssymb}  %
\usepackage{changepage} %
\usepackage{thmtools,thm-restate} %
\theoremstyle{plain}

\newcommand{\tuple}[1]{\langle {#1}\rangle}
\newcommand{\lang}[1]{{\mathcal{L}(#1)}} %
\newcommand{\preceqq}{\preccurlyeq}
\newcommand{\cA}{\mathcal{A}}
\newcommand{\cD}{\mathcal{D}}
\newcommand{\cF}[1]{\mathsf{Min}^{#1}}
\newcommand{\cG}[1]{\mathsf{Det}^{#1}}
\newcommand{\cH}{\mathsf{H}}
\newcommand{\cI}{\mathcal{I}}
\newcommand{\cN}{\mathcal{N}}
\newcommand{\cQ}{\mathcal{Q}}
\newcommand{\rl}{\sim^{\ell}}
\newcommand{\rr}{\sim^{r}}
\newcommand{\rlN}{\rl_{\cN}}
\newcommand{\rrN}{\rr_{\cN}}

\newcommand{\rlL}{\rl_{L}}
\newcommand{\rrL}{\rr_{L}}
\newcommand{\rc}{\curlywedge}
\newcommand{\ud}{\stackrel{\rm\scriptscriptstyle def}{=}}
\newcommand{\Lra}{\Leftrightarrow}
\newcommand{\Ra}{\Rightarrow}
\newcommand{\La}{\Leftarrow}
\newcommand{\ra}{\rightarrow}
\DeclareMathOperator{\pre}{pre}
\DeclareMathOperator{\post}{post}

\DeclareMathOperator{\gfp}{gfp}

\usepackage{tikz}
\usetikzlibrary{arrows}
\usetikzlibrary{positioning}
\usepackage{tikz-cd}
\usetikzlibrary{decorations.pathmorphing}

\usepackage{newunicodechar}
\newunicodechar{ε}{\varepsilon}
\newunicodechar{⟨}{\langle}
\newunicodechar{⟩}{\rangle}
\newunicodechar{⟦}{\llbracket}
\newunicodechar{⟧}{\rrbracket}
\newunicodechar{⌈}{\lceil}
\newunicodechar{⌉}{\rceil}
\newunicodechar{⇃}{\mathclose{\downharpoonleft}}
\newunicodechar{⋯}{\cdots}
\newunicodechar{∈}{\in}
\newunicodechar{∉}{\notin}
\newunicodechar{⊕}{\oplus}
\newunicodechar{⊗}{\otimes}
\newunicodechar{∘}{\circ}
\newunicodechar{α}{\alpha}
\newunicodechar{φ}{\varphi}
\newunicodechar{ψ}{\psi}
\newunicodechar{μ}{\mu}
\newunicodechar{π}{\pi}
\newunicodechar{τ}{\tau}
\newunicodechar{Σ}{\Sigma}
\newunicodechar{≤}{\leq}
\newunicodechar{≥}{\geq}
\newunicodechar{≝}{\stackrel{\rm\scriptscriptstyle def}{=}}
\newunicodechar{𝒪}{\mathcal{O}}
\newunicodechar{𝒯}{\mathcal{T}}
\newunicodechar{𝒴}{\mathcal{Y}}
\newunicodechar{Δ}{\Delta}

\usepackage{fancyvrb}
\VerbatimFootnotes

\usepackage[ruled,noend,noline,linesnumbered]{algorithm2e}
\makeatletter
\newcommand{\RemoveAlgoNumber}{\renewcommand{\fnum@algocf}{\AlCapSty{\AlCapFnt\algorithmcfname}}}
\makeatother

\nolinenumbers

\title{A Congruence-based Perspective on Automata Minimization Algorithms}
\titlerunning{A Congruence-based Perspective on Automata Minimization Algorithms} %

\author{Pierre Ganty}{IMDEA Software Institute, Madrid, Spain}{pierre.ganty@imdea.org}{0000-0002-3625-6003}
{Supported by the Spanish Ministry of Economy and Competitiveness project No. 
PGC2018-102210-B-I00, BOSCO -	Foundations for the development, analysis and understanding of BlOck chains and Smart COntracts, by the Madrid Regional Government project No. S2018/TCS-4339, BLOQUES -	Contratos inteligentes y Blockchains Escalables y Seguros mediante Verificación y Análisis, and by a Ramón y Cajal fellowship RYC-2016-20281.}

\author{Elena Gutiérrez}{IMDEA Software Institute, Madrid, Spain\\ Universidad Politécnica de Madrid, Spain}{elena.gutierrez@imdea.org}{0000-0001-5999-7608}{Supported by BES-2016-077136 grant from the Spanish Ministry of Economy, Industry and Competitiveness.}

\author{Pedro Valero}{IMDEA Software Institute, Madrid, Spain\\ Universidad Politécnica de Madrid, Spain}{pedro.valero@imdea.org}{0000-0001-7531-6374}{}

\authorrunning{P. Ganty and E. Gutiérrez and P.Valero} %

\Copyright{Pierre Ganty and Elena Gutiérrez and Pedro Valero}%



\category{}%

\supplement{}%

\EventEditors{Peter Rossmanith, Pinar Heggernes, and Joost-Pieter Katoen}
\EventNoEds{3}
\EventLongTitle{44th International Symposium on Mathematical Foundations of Computer Science (MFCS 2019)}
\EventShortTitle{MFCS 2019}
\EventAcronym{MFCS}
\EventYear{2019}
\EventDate{August 26--30, 2019}
\EventLocation{Aachen, Germany}
\EventLogo{}
\SeriesVolume{138}
\ArticleNo{50}

\begin{document}

\maketitle
\begin{abstract}
In this work we use a framework of finite-state automata constructions based on equivalences over words to provide new insights on the relation between well-known methods for computing the minimal deterministic automaton of a language.
\end{abstract}

\section{Introduction}
In this paper we consider the problem of building the minimal deterministic finite-state automaton generating a given regular language.
This is a classical issue that arises in many different areas of computer science such as verification, regular expression searching and natural language processing, to name a few.

There exists a number of methods, such as Hopcroft's~\cite{hopcroft1971} and Moore's algorithms~\cite{moore1956}, that receive as input a deterministic finite-state automaton (DFA for short) generating a language and build the minimal DFA for that language.
In general, these methods rely on computing a partition of the set of states of the input DFA which is then used as the set of states of the minimal DFA.

On the other hand, Brzozowski~\cite{brzozowski1962canonical} proposed the \emph{double-reversal method} for building the minimal DFA for the language generated by an input non-deterministic automaton (NFA for short).
This algorithm alternates a reverse operation and a determinization operation twice, relying on the fact that, for any given NFA \(\cN\), if the reverse automaton of \(\cN\) is deterministic then the determinization operation yields the minimal DFA for the language of \(\cN\).\linebreak
This method has been recently generalized by Brzozowski and Tamm~\cite{brzozowski2014theory}.
They showed the following \emph{necessary and sufficient condition}: the determinization operation yields the minimal DFA for the language of \(\cN\) if and only if the reverse automaton of \(\cN\) is \emph{atomic}.

It is well-known that all these approaches to the DFA minimization problem aim to compute Nerode's equivalence relation for the considered language.
However, the double-reversal method and its later generalization appear to be quite isolated from other methods such as Hopcroft's and Moore's algorithms.
This has led to different attempts to better explain Brzozowski's method~\cite{bonchi_algebra-coalgebra_2014} and its connection with other minimization algorithms~\cite{Adamek2012,Champarnaud2002,Garcia2013}.
We use a framework of automata constructions based on equivalence classes over \emph{words} to give new insights on the relation between these algorithms.

In this paper we consider equivalence relations over words on an alphabet \(\Sigma\) that induce finite partitions over \(\Sigma^*\).
Furthermore, we require that these partitions are well-behaved with respect to concatenation, namely, \emph{congruences}.
Given a regular language \(L\) and an equivalence relation satisfying these conditions, we use well-known automata constructions that yield automata generating the language \(L\)~\cite{Buchi89,Khoussainov2001}.
In this work, we consider two types of equivalence relations over words verifying the required conditions.

First, we define a \emph{language-based equivalence}, relative to a regular language, that behaves well with respect to \emph{right} concatenation, also known as the right Nerode's equivalence relation for the language.
When applying the automata construction to the right Nerode's equivalence, we obtain the minimal DFA for the given language~\cite{Buchi89,Khoussainov2001}.
In addition, we define an \emph{automata-based equivalence}, relative to an NFA.
When applying the automata construction to the automata-based equivalence we obtain a determinized version of the input NFA.

On the other hand, we also obtain counterpart automata constructions for relations that are well-behaved with respect to \emph{left }concatenation.
In this case, language-based and automata-based equivalences yield, respectively, the minimal co-deterministic automaton and a co-deterministic NFA for the language.

The relation between the automata constructions resulting from the language-based and the automata-based congruences, together with the the duality between right and left congruences, allows us to relate determinization and minimization operations.
As a result, we formulate a sufficient and necessary condition that guarantees that determinizing an automaton yields the minimal DFA.
This formulation evidences the relation between the double-reversal and the state partition refinement minimization methods.

We start by giving a simple proof of Brzozowski's double-reversal method~\cite{brzozowski1962canonical}, to later address the generalization of Brzozowski and Tamm~\cite{brzozowski2014theory}.
Furthermore, we relate the iterations of Moore's partition refinement algorithm, which works on the states of the input DFA, to the iterations of the greatest fixpoint algorithm that builds the right Nerode's partition on words.
We conclude by relating the automata constructions introduced by Brzozowski and Tamm~\cite{brzozowski2014theory}, named the \emph{átomaton} and \emph{the partial átomaton}, to the automata constructions described in this work.

\subparagraph{Structure of the paper.} %
\label{subp:structure}
After preliminaries in Section~\ref{sec:preliminaries}, we introduce in Section~\ref{sec:automataConstructions} the automata constructions based on congruences on words and establish the duality between these constructions when using right and left congruences.
Then, in Section~\ref{sec:congruences}, we define language-based and automata-based congruences and analyze the relations between the resulting automata constructions.
In Section~\ref{sec:Novel}, we study a collection of well-known constructions for the minimal DFA.
Finally, we give further details on related work in Section~\ref{sec:relatedWork}. %
For space reasons, missing proofs are deferred to the Appendix.

\section{Preliminaries}
\label{sec:preliminaries}

\subparagraph{Languages.}
Let \(\Sigma\) be a finite nonempty \emph{alphabet} of symbols.
Given a word \(w \in \Sigma^*\), \(w^R\) denotes the \emph{reverse} of \(w\).
Given a language \(L \subseteq \Sigma^*\), \(L^R ≝ \{w^R \mid w \in L\}\) denotes the \emph{reverse language} of \(L\).
We denote by \(L^c\) the \emph{complement} of the language \(L\).
The \emph{left (resp. right) quotient} of L by a word \(u\) is defined as the language \(u^{-1}L ≝ \{x \in \Sigma^* \mid ux \in L\}\) (resp. \(Lu^{-1} ≝ \{x \in \Sigma^* \mid xu \in L\}\)).

\subparagraph{Automata.}
A  \emph{(nondeterministic) finite-state automaton} (NFA for short), or simply \emph{automaton}, is a 5-tuple \(\cN = (Q, \Sigma, \delta, I, F)\), where \(Q\) is a finite set of \emph{states}, \(\Sigma\) is an alphabet, \({I\subseteq Q}\) are the \emph{initial} states, \(F \subseteq Q\) are the \emph{final} states, and \(\delta: Q \times \Sigma \ra \wp(Q)\) is the \emph{transition} function.
We denote the \emph{extended transition function} from \(\Sigma\) to \(\Sigma^*\) by \(\hat{\delta}\).
Given \(S,T \subseteq Q\), \(W^{\cN}_{S,T} \ud \{w \in \Sigma^* \mid \exists q \in S, q' \in T: q' \in \hat{\delta}(q,w)\}\).
In particular, when \(S = \{q\}\) and \(T = F\), we define \emph{the right language} of state \(q\) as \(W^{\cN}_{q,F}\).
Likewise, when \(S = I\) and \(T = \{q\}\), we define the \emph{left language} of state \(q\) as \(W^{\cN}_{I,q}\).
We define \(\post_w^{\cN}(S) \ud \{q \in Q \mid w \in W^{\cN}_{S,q}\}\) and \(\pre_w^{\cN}(S) \ud \{q \in Q \mid w \in W^{\cN}_{q,S}\}\).
In general, we omit the automaton \(\cN\) from the superscript when it is clear from the context.
We say that a state \(q\) is \emph{unreachable} if{}f \(W^{\cN}_{I,q} = \emptyset\) and we say that \(q\) is \emph{empty} if{}f \(W^{\cN}_{q,F} = \emptyset\).
Finally, note that \(\lang{\cN} = \bigcup_{q \in I} W_{q,F}^{\cN} = \bigcup_{q \in F} W_{I,q}^{\cN} = W_{I,F}^{\cN}\).

Given an NFA \(\cN = (Q, \Sigma, \delta, I, F)\), the \emph{reverse NFA} for \(\cN\), denoted by \(\cN^R\), is defined as \(\cN^R = (Q, \Sigma, \delta_r, F, I)\) where \(q \in \delta_r (q',a)\) if{}f \(q' \in \delta(q,a)\).
Clearly, \(\lang{\cN}^R = \lang{\cN^R}\).

A \emph{deterministic finite-state automaton} (DFA for short) is an NFA such that, \(I = \{q_0\}\), and, for every state \(q \in Q\) and every symbol \(a \in \Sigma\), there exists exactly one \(q' \in Q\) such that \(\delta(q,a) = q'\).
According to this definition, DFAs are always \emph{complete}, i.e., they define a transition for each state and input symbol.
In general, we denote NFAs by \(\cN\), using \(\cD\) for DFAs when the distinction is important.
A \emph{co-deterministic finite-state automata} (co-DFA for short) is an NFA \(\cN\) such that \(\cN^R\) is deterministic.
In this case, co-DFAs are always \emph{co-complete}, i.e., for each target state \(q'\) and each input symbol, there exists a source state \(q\) such that \(\delta(q,a) = q'\).
Recall that, given an NFA \(\cN = (Q, \Sigma, \delta, I, F)\), the well-known \emph{subset construction} builds a DFA \(\cD = (\wp(Q), \Sigma, \delta_{d}, \{I\}, F_{d})\) where \(F_d = \{S \in \wp(Q) \mid S \cap F \neq \emptyset\}\) and \(\delta_d(S,a) = \{q' \mid \exists q \in S, q' \in \delta(q,a)\}\) for every \(a \in \Sigma\), that accepts the same language as \(\cN\)~\cite{Ullman2003}.
Given an NFA \(\cN = (Q, \Sigma, \delta, I, F)\), we denote by \(\cN^D\) the DFA that results from applying the subset construction to \(\cN\) where only subsets (including the empty subset) that are reachable from the initial subset of \(\cN^D\) are used.
Then, \(\cN^D\) possibly contains empty states but no state is unreachable.
A DFA for the language \(\lang{\cN}\) is \emph{minimal}, denoted by \(\cN^{DM}\), if it has no unreachable states and no two states have the same right language.
The minimal DFA for a regular language is unique modulo isomorphism.

\subparagraph{Equivalence Relations and Partitions.}
Recall that an \emph{equivalence relation} on a set \(X\) is a binary relation \(\sim\) that is reflexive, symmetric and transitive.
Every equivalence relation \(\sim\) on \(X\) induces a \emph{partition} \(P_{\sim}\) of \(X\), i.e., a family \(P_{\sim} = \{B_i\}_{i\in \cI}\subseteq \wp(X)\) of subsets of \(X\), with \(\cI \subseteq \mathbb{N}\), such that: 
\begin{romanenumerate}%
\item \(B_i \neq \emptyset\) for all \(i \in \cI\);
\item \(B_i \cap B_j = \emptyset\), for all \(i,j \in \cI\) with \(i \neq j\); and
\item \(X = \bigcup_{i \in \cI} B_i\). 
\end{romanenumerate}

We say that a partition is \emph{finite} when \(\cI\) is finite.
Each \(B_i\) is called a \emph{block} of the partition.
Given \(u\in X\), then \(P_{\sim}(u)\) denotes the unique block that contains \(u\) and corresponds to the \emph{equivalence class} \(u\) w.r.t. \(\sim\), \( P_{\sim}(u) ≝ \{v \in X \mid u \sim v \}\).
This definition can be extended in a natural way to a set \(S \subseteq X\) as \(P_{\sim}(S) ≝ \bigcup_{u \in S} P_{\sim}(u)\).
We say that the partition \(P_{\sim}\) \emph{represents precisely} \(S\) if{}f \(\ P_{\sim}(S) = S\).
An equivalence relation \(\sim\) is \emph{of finite index} if{}f \(\sim\) defines a finite number of equivalence classes, i.e., the induced partition \(P_{\sim}\) is finite.
In the following, we will always consider equivalence relations of finite index, i.e., finite partitions.

Finally, denote \(Part(X)\) the set of partitions of \(X\).
We use the standard refinement ordering \(\preceq\) between partitions: let \(P_1, P_2 \in Part(X)\), then \(P_1 \preceq P_2\) if{}f for every \(B \in P_1\) there exists \(B' \in P_2\) such that \(B \subseteq B'\).
Then, we say that \(P_1\) is \emph{finer than} \(P_2\) (or equivalently, \(P_2\) is \emph{coarser than} \(P_1\)).
Given \(P_1, P_2 \in Part(X)\), define the \emph{coarsest common refinement}, denoted by \(P_1 \curlywedge P_2\), as the coarsest partition \(P \in Part(X)\) that is finer than both \(P_1\) and \(P_2\).
Likewise, define the \emph{finest common coarsening}, denoted by \(P_1 \curlyvee P_2\), as the finest partition \(P\) that is coarser than both \(P_1\) and \(P_2\).
Recall that \((Part(X), \preceq, \curlyvee, \curlywedge)\) is a complete lattice where the top (coarsest) element is \(\{X\}\) and the bottom (finest) element is \(\{\{x\}\mid x \in X\}\).

\section{Automata Constructions from Congruences}\label{sec:automataConstructions}
We will consider equivalence relations on \(\Sigma^*\) (and their corresponding partitions) with good properties w.r.t. concatenation.
An equivalence relation \(\sim\) is a \emph{right (resp. left) congruence} if{}f for all \(u,v \in Σ^*\), we have that \(u \sim v \Ra ua\sim va\), for all \(a \in \Sigma\) (resp. \(u \sim v \Ra au \sim av\)). 
We will denote right congruences (resp. left congruences) by \(\rr\) (resp. \(\rl\)).
The following lemma gives a characterization of right and left congruences.

\begin{restatable}{lemmaR}{CongruencePbwComplete}\label{lemma:CongruencePbwComplete}
The following properties hold:
\begin{enumerate}
\item \(\rr\) is a right congruence if{}f \(P_{\rr}(v)u \subseteq P_{\rr}(vu)\), for all \(u,v \in \Sigma^*\).
\item\(\rl\) is a left congruence if{}f \(uP_{\rl}(v) \subseteq P_{\rl}(uv)\), for all \(u,v \in \Sigma^*\).
\end{enumerate}
\end{restatable}

\vspace{4pt}

Given a right congruence \(\rr\) and a regular language \(L \subseteq \Sigma^*\) such that \(P_{\rr}\) represents precisely \(L\), i.e., \(P_{\rr}(L) = L\), the following automata construction recognizes exactly the language \(L\)~\cite{Khoussainov2001}.

\begin{definition}[Automata construction \(\cH^{r}(\rr, L)\)]
\label{def:right-const}
Let \(\rr\) be a right congruence and let \(P_{\rr}\) be the partition induced by \(\rr\).
Let \(L \subseteq \Sigma^*\) be a language.
Define the automaton \(\cH^{r}(\rr, L)=\linebreak (Q, \Sigma, \delta, I, F)\) where \(Q = \{P_{\rr}(u) \mid u \in \Sigma^*\}\), \(I = \{P_{\rr}(\varepsilon) \}\), \(F = \{P_{\rr}(u) \mid u \in L\}\), and \(\delta(P_{\rr}(u), a) = P_{\rr}(v)\) if{}f \(P_{\rr}(u)a \subseteq P_{\rr}(v)\), for all \(u,v \in \Sigma^*\) and \(a \in \Sigma\).
\end{definition}

\begin{remark} 
\label{remark:deterministic}
Note that \(\cH^{r}(\rr, L)\) is \emph{finite} since we assume \(\rr\) is of finite index.
Note also that \(\cH^{r}(\rr, L)\) is a complete \emph{deterministic} finite-state automaton since, for each \(u \in \Sigma^*\) and \(a \in \Sigma\), there exists \emph{exactly one} block \(P_{\rr}(v)\) such that \(P_{\rr}(u)a \subseteq P_{\rr}(v)\), which is \(P_{\rr}(ua)\).
Finally, observe that \(\cH^{r}(\rr, L)\) possibly contains empty states but no state is unreachable.
\end{remark}

\begin{restatable}{lemmaR}{HrGeneratesL}
\label{lemma: HrGeneratesL}
Let \(\rr\) be a right congruence and let \(L\subseteq \Sigma^*\) be a language such that \(P_{\rr}(L) = L\).
Then \(\lang{\cH^{r}(\rr, L)} = L\).
\end{restatable}

Due to the left-right duality between \(\rl\) and \(\rr\), we can give a similar automata construction such that, given a left congruence \(\rl\) and a language \(L \subseteq \Sigma^*\) with \(P_{\rl}(L) = L\), recognizes exactly the language \(L\).

\begin{definition}[Automata construction \(\cH^{\ell}(\rl, L)\)]
\label{def:left-const}
Let \(\rl\) be a left congruence and let \(P_{\rl}\) be the partition induced by \(\rl\).
Let \(L \subseteq \Sigma^*\) be a language.
Define the automaton \(\cH^{\ell}(\rl, L)=\linebreak (Q, \Sigma, \delta, I, F)\) where \(Q = \{P_{\rl}(u) \mid u \in \Sigma^*\}\), \(I = \{ P_{\rl}(u) \mid u \in L \}\), \(F = \{P_{\rl}(\varepsilon)\}\), and \( P_{\rl}(v) \in \delta(P_{\rl}(u), a)\)  if{}f \(aP_{\rl}(v) \subseteq P_{\rl}(u)\), for all \(u,v \in \Sigma^*\) and \(a \in \Sigma\).
\end{definition}

\begin{remark} 
\label{remark:co-deterministic}
In this case, \(\cH^{\ell}(\rl, L)\) is a co-complete \emph{co-deterministic} finite-state automaton since, for each \(v \in \Sigma^*\) and \(a \in \Sigma\), there exists \emph{exactly one} block \(P_{\rl}(u)\) such that \linebreak\(aP_{\rl}(v) \subseteq P_{\rl}(u)\), which is \(P_{\rl}(av)\).
Finally, observe that \(\cH^{\ell}(\rl, L)\) possibly contains unreachable states but no state is empty.
\end{remark}

\begin{restatable}{lemmaR}{HlgeneratesL}
\label{lemma:HlgeneratesL}
Let \(\rl\) be a left congruence and let \(L\subseteq \Sigma^*\) be a language such that \(P_{\rl}(L) = L\).
Then \(\lang{\cH^{\ell}(\rl, L)} = L\).
\end{restatable}

Lemma~\ref{lemma:AutomataPartitionleftRightReverse} shows that \(\cH^\ell\) and \(\cH^r\) inherit the left-right duality between \(\rl\) and \(\rr\).

\begin{restatable}{lemmaR}{AutomataPartitionleftRightReverse}
\label{lemma:AutomataPartitionleftRightReverse}
Let \(\rr\) and \(\rl\) be a right and left congruence respectively, and let \(L \subseteq \Sigma^*\) be a language.
If the following property holds
\begin{equation}\label{eq:AutomataPartitionleftRightReverse}
u \rr v \Lra u^R \rl v^R
\end{equation}
then \(\cH^{r}(\rr, L) \) is isomorphic to \( \left(\cH^{\ell}(\rl, L^R)\right)^R\).
\end{restatable}

\section{Language-based Congruences and their Approximation using NFAs}\label{sec:congruences}
\label{sec: Languages}

Given a language \(L \subseteq \Sigma^*\), we recall the following equivalence relations on \(\Sigma^*\), which are often denoted as \emph{Nerode's equivalence relations} (e.g., see \cite{Khoussainov2001}).

\begin{definition}[Language-based Equivalences]\label{def:language}
Let \(u,v \in \Sigma^*\) and let \(L \subseteq \Sigma^*\) be a language. 
Define:
\begin{align}
u \sim^{r}_{L} v  & \Lra u^{-1}L = v^{-1}L & \quad \text{\emph{Right-}language-based Equivalence}\label{eq:Rlanguage} \\
u \sim^{\ell}_{L} v  & \Lra Lu^{-1} = Lv^{-1}  & \quad \text{\emph{Left-}language-based Equivalence} \label{eq:Llanguage}
\end{align}
\end{definition}

Note that the right and left language-based equivalences defined above are, respectively, right and left \emph{congruences} (for a proof, see Lemma~\ref{lemma:lang-cong} in the Appendix).
Furthermore, when \(L\) is a regular language, \(\rrL\) and \(\rlL\) are of \emph{finite index} \cite{Buchi89,Khoussainov2001}.
Since we are interested in congruences of finite index (or equivalently, finite partitions), we will always assume that \(L\) is a regular language over \(\Sigma\).

The following result states that, given a language \(L\), the right Nerode's equivalence induces the coarsest partition of \(\Sigma^*\) which is a right congruence and precisely represents \(L\).
\begin{lemma}[de Luca and Varricchio~\cite{deluca2011}]
\label{lemma:Varricchio}
Let \(L\subseteq \Sigma^*\) be a regular language. Then,
\[P_{\rrL} = \bigcurlyvee \{P_{\rr} \mid \; \rr \text{is a right congruence and } P_{\rr}(L) = L\} \enspace .\]
\end{lemma}

In a similar way, one can prove that the same property holds for the left Nerode's equivalence.
Therefore, as we shall see, applying the construction \(\cH\) to these equivalences yields minimal automata.
However, computing them becomes unpractical since languages are possibly infinite, even if they are regular.
Thus, we will consider congruences based on the states of the NFA-representation of the language which induce finer partitions of \(\Sigma^*\) than Nerode's equivalences.
In this sense, we say that the automata-based equivalences \emph{approximate} Nerode's equivalences.

\begin{definition}[Automata-based Equivalences] \label{def:automataEquiv}
Let \(u,v \in\Sigma^*\) and let \(\cN = (Q, \Sigma, \delta, I, F)\) be an NFA.
Define:
\begin{align}
u \rrN v & \Lra \post^{\cN}_u(I) = \post^{\cN}_v(I) & \quad \text{\emph{Right-}automata-based Equivalence}\label{eq:RState} \\
u \rlN v & \Lra \pre^{\cN}_u(F) = \pre^{\cN}_v(F) & \quad \text{\emph{Left-}automata-based Equivalence} \label{eq:LState} 
\end{align}
\end{definition}

Note that the right and left automata-based equivalences defined above are, respectively, right and left \emph{congruences} (for a proof, see Lemma~\ref{lemma:automata-cong} in the Appendix).
Furthermore, they are of \emph{finite index} since each equivalence class is represented by a subset of states of \(\cN\).

The following result gives a sufficient and necessary condition for the language-based (Definition~\ref{def:language}) and the automata-based equivalences (Definition~\ref{def:automataEquiv}) to coincide.

\begin{restatable}{lemmaR}{languageEqualPost}
\label{lemma:languageEqualPost}
Let \(\cN = (Q,\Sigma,\delta,I,F)\) be an automaton with \(L = \lang{\cN}\).
Then,
\begin{equation}\label{eq:RightSetsDisjoint}
\mathord{\rrL} = \mathord{\rrN} \enspace \text{ if{}f } \enspace \forall u,v \in \Sigma^*, \; W_{\post_{u}^{\cN}(I),F}^{\cN} = W_{\post_{v}^{\cN}(I),F}^{\cN} \Lra \post_{u}^{\cN}(I) = \post_{v}^{\cN}(I)\enspace  .
\end{equation}

\end{restatable}

\subsection{Automata Constructions}

In what follows, we will use \(\cF{}\) and \(\cG{}\) to denote the construction \(\cH\) when applied, respectively, to the language-based congruences induced by a regular language and the automata-based congruences induced by an NFA.

\begin{definition}\label{def:FG}
Let \(\cN\) be an NFA generating the language \(L = \lang{\cN}\).
Define:
\begin{align*}
\cF{r}(L) & ≝  \cH^{r}(\rrL, L) & \cG{r}(\cN) & ≝ \cH^{r}(\rrN, L) \\
\cF{\ell}(L) & ≝  \cH^{\ell}(\rlL, L) & \cG{\ell}(\cN) & ≝  \cH^{\ell}(\rlN, L) \enspace .
\end{align*}
\end{definition}

Given an NFA \(\cN\) generating the language \(L=\lang{\cN}\), all constructions in the above definition yield automata generating \(L\).
However, while the constructions using the right congruences result in DFAs, the constructions relying on left congruences result in co-DFAs.
Furthermore, since the pairs of relations~\eqref{eq:Rlanguage}-\eqref{eq:Llanguage} and~\eqref{eq:RState}-\eqref{eq:LState}, from Definition~\ref{def:language} and~\ref{def:automataEquiv} respectively, are dual, i.e., they satisfy the hypothesis of Lemma~\ref{lemma:AutomataPartitionleftRightReverse}, it follows that \(\cF{\ell}(L)\) is isomorphic to \((\cF{r}(L^R))^R\) and \(\cG{\ell}(\cN)\) is isomorphic to \((\cG{r}(\cN^R))^R\).

On the other hand, since \(\cF{r}\) relies on the language-based congruences, the resulting DFA is minimal, which is not guaranteed to occur with \(\cG{r}\).
This easily follows from the fact that the states of the automata constructions are the equivalence classes of the given congruences and there is no right congruence (representing \(L\) precisely)  that is coarser than the right Nerode's equivalence (see Lemma \ref{lemma:Varricchio}).

Finally, since every co-deterministic automaton satisfies the right-hand side of Equation~\eqref{eq:RightSetsDisjoint}, it follows that determinizing (\(\cG{r}\)) a co-deterministic automaton (\(\cG{\ell}(\cN)\)) results in the minimal DFA (\(\cF{r}(\lang{\cN})\)), as already proved by Sakarovitch~\cite[Proposition 3.13]{sakarovitch2009elements}.

We formalize all these notions in Theorem~\ref{theoremF}.
Finally, Figure~\ref{Figure:diagramAutomata} summarizes all these well-known connections between the automata constructions given in Definition \ref{def:FG}.

\begin{restatable}{theoremR}{theoremF}\label{theoremF}
Let \(\cN\) be an NFA generating language \(L = \lang{\cN}\).
Then the following properties hold:
\begin{alphaenumerate}%
\item \(\lang{\cF{r}(L)} =\lang{\cF{\ell}(L)} = L = \lang{\cG{r}(\cN)} = \lang{\cG{\ell}(\cN)}\). \label{lemma:language-F}
\item \(\cF{r}(L)\) is isomorphic to the minimal deterministic automaton for \(L\).\label{theorem:minimalDFAlanguage}
\item \(\cG{r}(\cN)\) is isomorphic to \(\cN^{D}\).\label{lemma:rightNDet}
\item \(\cF{\ell}(L)\) is isomorphic to \((\cF{r}(L^R))^R\). \label{lemma:FlisomorphicRfrR}
\item \(\cG{\ell}(\cN)\) is isomorphic to \((\cG{r}(\cN^R))^R\).\label{lemma:AlRequalArNR}
\item \(\cG{r}(\cG{\ell}(\cN))\) is isomorphic to \(\cF{r}(L)\).\label{lemma:LS+RS=RN}
\end{alphaenumerate}
\end{restatable}

\begin{figure}[!ht]
\begin{minipage}[l]{0.42\textwidth}
\begin{tikzcd}[column sep=normal, row sep=large]
\cN \ar[d, description, "R"',leftrightarrow] \ar[r, "\cG{\ell}"] \ar[rr, bend left=30, "\cF{r}"] & \cG{\ell}(\cN) \ar[d, description, "R"',leftrightarrow] \ar[r, "\cG{r}"] & \cG{r}(\cG{\ell}(\cN)) \ar[d, description, "R"',leftrightarrow]\\
\cN^R \ar[r, "\cG{r}"] \ar[rr, bend right=30, "\cF{\ell}"] & \cG{r}(\cN^R) \ar[r, "\cG{\ell}"] & \cG{\ell}(\cG{r}(\cN^R))
\end{tikzcd}
\end{minipage}%
\hspace{1.2cm}
\begin{minipage}[r]{0.48\textwidth}
The upper part of the diagram follows from Theorem~\ref{theoremF}~(\ref{lemma:LS+RS=RN}).
Both squares of the diagram fo\-llow from Theorem~\ref{theoremF}~(\ref{lemma:AlRequalArNR}), which states that \(\cG{\ell}(\cN)\) is isomorphic to \((\cG{r}(\cN^R))^R\).
Finally, the bottom curved arc follows from 
Theorem~\ref{theoremF}~(\ref{lemma:FlisomorphicRfrR}).
Incidentally, the diagram shows a new relation which follows from the left-right dualities between \(\rlL\) and \(\rrL\), and \(\rlN\) and \(\rrN\): \(\cF{\ell}(\lang{\cN^R})\) is isomorphic to \(\cG{\ell}(\cG{r}(\cN^R))\).
\end{minipage}

\caption{Relations between the constructions \(\cG{\ell},\cG{r},\cF{\ell}\) and \(\cF{r}\).
Note that constructions \(\cF{r}\) and \(\cF{\ell}\) are applied  to the language generated by the automaton in the origin of the labeled arrow, while constructions \(\cG{r}\) and \(\cG{\ell}\) are applied directly to the automaton.}
\label{Figure:diagramAutomata}
\end{figure}

\section{A Congruence-based Perspective on Known Algorithms}
\label{sec:Novel}
We can find in the literature several well-known independent techniques for the construction of minimal DFAs.
Some of these techniques are based on refining a state partition of an input DFA, such as Moore's algorithm~\cite{moore1956}, while others directly manipulate an input NFA, such as the double-reversal method~\cite{brzozowski1962canonical}.
Now, we establish a connection between these algorithms through Theorem~\ref{theorem:minimalifreverseatomic}, which gives a necessary and sufficient condition on an NFA so that determinizing it yields the minimal DFA.

\begin{lemma}\label{lemma:minimalifreverseatomic}
Let \(\cN = (Q,\Sigma,\delta,I,F)\) be an NFA with \(L=\lang{\cN}\) and \(\rrL=\rrN\).
Then \(\forall q \in Q,\;  P_{\rrL}(W_{I,q}^{\cN}) = W_{I,q}^{\cN}\).
\end{lemma}
\begin{proof}
\begin{align*}
P_{\rrL}(W_{I,q}^{\cN}) & = \quad \text{[By definition of \(P_{\rrL}\)]}\\
\{w \in \Sigma^* \mid \exists u \in W_{I,q}^{\cN}, \; w^{-1}L = u^{-1}L\} & = \quad \text{[Since \(\rrL = \rrN\)]} \\
\{w \in \Sigma^* \mid \exists u \in W_{I,q}^{\cN}, \; \post_w^{\cN}(I) = \post_u^{\cN}(I)\} & \subseteq  \quad \text{[\( u \in W_{I,q}^{\cN} \iff q \in \post_u^{\cN}(I)\)]}\\
\{w \in \Sigma^* \mid q \in \post_w^{\cN}(I)\} & = \quad \text{[By definition of \(W_{I,q}^{\cN}\)]}\\
W_{I,q}^{\cN} \enspace .
\end{align*}
By reflexivity of \(\rrL, \) we conclude that \(P_{\rrL}(W_{I,q}^{\cN}) = W_{I,q}^{\cN} \).
\end{proof}

\begin{theorem}\label{theorem:minimalifreverseatomic}
Let \(\cN = (Q,\Sigma,\delta,I,F)\) be an NFA with \(L=\lang{\cN}\).
Then \(\cG{r}(\cN)\) is the minimal DFA for \(L\) if{}f \(\forall q \in Q,\;  P_{\rrL}(W_{I,q}^{\cN}) = W_{I,q}^{\cN}\).
\end{theorem}
\begin{proof}
Assume \(\cG{r}(\cN)\) is minimal. 
Then \(P_{\rrN}(u) = P_{\rrL}(u)\) for all \(u \in \Sigma^*\), i.e. \(\mathord{\rrL} = \mathord{\rrN}\).
It follows from Lemma~\ref{lemma:minimalifreverseatomic} that \(P_{\rrL}(W_{I,q}^{\cN}) = W_{I,q}^{\cN} \).

Now, assume that \(P_{\rrL}(W_{I,q}^{\cN}) = W_{I,q}^{\cN}\), for each \(q \in Q\).
Then, for every \(u \in \Sigma^*\),
\[P_{\rrN}(u) = \hspace{-5pt}\bigcap\limits_{\substack{q \in \post^{\cN}_u(I)}} W_{I,q}^{\cN} \; \cap \hspace{-5pt}\bigcap\limits_{\substack{q \notin \post^{\cN}_u(I)}} (W_{I,q}^{\cN})^c = \hspace{-5pt}\bigcap\limits_{\substack{q \in \post^{\cN}_u(I)}} P_{\rrL}(W_{I,q}^{\cN}) \;\; \cap \hspace{-5pt}\bigcap\limits_{\substack{q \notin \post^{\cN}_u(I)}} (P_{\rrL}(W_{I,q}^{\cN}))^c\]
where the first equality follows by rewriting \(P_{\rrN}(u) = \{v \in \Sigma^* \mid \post^{\cN}_u(I) = \post^{\cN}_v(I)\}\) with universal quantifiers, hence intersections, and the last equality follows from the initial assumption \(P_{\rrL}(W_{I,q}^{\cN}) = W_{I,q}^{\cN} \).

It follows that \(P_{\rrN}(u)\) is a \emph{union} of blocks of \(P_{\rrL}\).
Recall that \(\rrL\) induces the coarsest right congruence such that \(P_{\rrL}(L)=L\) (Lemma \ref{lemma:Varricchio}).
Since \(\rrN\) is a right congruence satisfying \(P_{\rrN}(L)=L\) then \(P_{\rrN} \preceqq P_{\rrL}\).
Therefore, \(P_{\rrN}(u)\) necessarily corresponds to one single block of \(P_{\rrL}\), namely, \(P_{\rrL}(u)\).
Since \(P_{\rrN}(u) = P_{\rrL}(u)\) for each \(u \in \Sigma^*\), we conclude that \(\cG{r}(\cN) = \cF{r}(L)\).
\end{proof}

\subsection{Double-reversal Method}
In this section we give a simple proof of the well-known double-reversal minimization algorithm of Brzozowski~\cite{brzozowski1962canonical} using Theorem~\ref{theorem:minimalifreverseatomic}.
Note that, since \(\cG{r}(\cN)\) is isomorphic to \(\cN^D\) by Theorem~\ref{theoremF}~(\ref{lemma:rightNDet}), the following result coincides with that of Brzozowski.

\begin{theorem}[\cite{brzozowski1962canonical}]\label{theorem:DoubleReversal}
Let \(\cN\) be an NFA.
Then \(\cG{r}((\cG{r}(\cN^R))^R)\) is isomorphic to the minimal DFA for \(\lang{\cN}\).
\end{theorem}

\begin{proof}
Let \(L = \lang{\cN}\).
By definition, \(\cN' = (\cG{r}(\cN^R))^R\) is a co-DFA and, therefore, satisfies the condition on the right-hand side of Equation~\eqref{eq:RightSetsDisjoint}.
It follows from Lemma~\ref{lemma:languageEqualPost} that \(\rrL=\rr_{\cN'}\) which, by Lemma~\ref{lemma:minimalifreverseatomic} and Theorem~\ref{theorem:minimalifreverseatomic}, implies that \(\cG{r}(\cN')\) is minimal.
\end{proof} 

Note that Theorem~\ref{theorem:DoubleReversal} can be inferred from Figure~\ref{Figure:diagramAutomata} by following the path starting at \(\cN\), labeled with \(R-\cG{r}-R-\cG{r}\) and ending in \(\cF{r}(\lang{\cN})\).

\subsection{Generalization of the Double-reversal Method}
Brzozowski and Tamm~\cite{brzozowski2014theory} generalized the double-reversal algorithm by defining a necessary and sufficient condition on an NFA which guarantees that the determinized automaton is minimal.
They introduced the notion of \emph{atomic} NFA and showed that \(\cN^D\) is minimal if{}f \(\cN^R\) is atomic.
We shall show that this result is equivalent to Theorem~\ref{theorem:minimalifreverseatomic} due to the left-right duality between the language-based equivalences (Lemma~\ref{lemma:AutomataPartitionleftRightReverse}).

\begin{definition}[Atom~\cite{brzozowski2014theory}]\label{def:atom}
Let \(L\) be a regular language \(L\).
Let \(\{K_i \mid 0 \leq i \leq n-1\}\) be the set of left quotients of \(L\).
An \emph{atom} is any non-empty intersection of the form \(\widetilde{K_0} \cap \widetilde{K_1} \cap \ldots \cap \widetilde{K_{n-1}}\), where each \(\widetilde{K_i}\) is either \(K_i\) or \(K_i^c\).
\end{definition}

This notion of atom coincides with that of equivalence class for the left language-based congruence \(\rlL\).
This was first noticed by Iván~\cite{Ivan16}.

\begin{restatable}{lemma}{atoms}\label{eq:atoms}
Let \(L\) be a regular language.  
Then for every \(u \in \Sigma^*\),
\[P_{\rlL}(u) = \bigcap\limits_{\substack{u \in w^{-1}L \\w \in \Sigma^* }} w^{-1}L \; \cap \bigcap\limits_{\substack{u \notin w^{-1}L \\w \in \Sigma^* }} (w^{-1}L)^c \enspace .\]
\end{restatable}

\begin{definition}[Atomic NFA~\cite{brzozowski2014theory}]\label{def:atomicNFA}
An NFA \(\cN = (Q,\Sigma,\delta,I,F)\) is \emph{atomic} if{}f for every state \(q \in Q\), the right language \(W_{q,F}^{\cN}\) is a union of atoms of \(\lang{\cN}\).
\end{definition}

It follows from Lemma~\ref{eq:atoms} that the set of atoms of a language \(L\) corresponds to the partition \(P_{\rlL}\).
Therefore, a set \(S \subseteq \Sigma^*\) is a union of atoms if{}f \(P_{\rlL}(S)=S\).
This property, together\linebreak with Definition~\ref{def:atomicNFA}, shows that an NFA \(\cN = (Q,\Sigma,\delta,I,F)\) with \(L = \lang{\cN}\) is \emph{atomic} if{}f 
\begin{equation}\label{eq:atomicNFAClosed}
\forall q \in Q, \; P_{\rlL}(W_{q,F}^{\cN}) = W_{q,F}^{\cN} \enspace .
\end{equation}

We are now in condition to give an alternative proof of the generalization of Brzozowski and Tamm~\cite{brzozowski2014theory} relying on Theorem~\ref{theorem:minimalifreverseatomic}.

\begin{lemma}\label{lemma:atomiccoatomic}
Let \(\cN=(Q,\Sigma,\delta,I,F)\) be an NFA with \(L=\lang{\cN}\).
Then \(\cN^R\) is atomic if{}f \(\cG{r}(\cN)\) is the minimal DFA for \(L\).
\end{lemma}

\begin{proof}
Let \(\cN^R=(Q,\Sigma,\delta_r,F,I)\) and \(L^R = \lang{\cN^R}\).
Then,
\begin{align*}
\forall q \in Q, \; P_{\sim^{\ell}_{L^R}}(W_{q,I}^{\cN^R}) = W_{q,I}^{\cN^R} &\iff \quad  \text{[By \(A = B \Lra A^R = B^R\)]} \\
\forall q \in Q, \; \left(P_{\sim^{\ell}_{L^R}}(W_{q,I}^{\cN^R})\right)^R = \left(W_{q,I}^{\cN^R}\right)^R &\iff \quad  \text{[By \(u\rlL v \Lra u^R \sim^r_{L^R} v^R\)]} \\
\forall q \in Q, \; P_{\rrL}\left(\left(W_{q,I}^{\cN^R}\right)^R\right) = \left(W_{q,I}^{\cN^R}\right)^R &\iff \quad  \text{[By \(\left(W^{\cN^R}_{q,I}\right)^R = W_{I,q}^{\cN}\)]} \\
\forall q \in Q, \; P_{\rrL}(W_{I,q}^{\cN}) = W_{I,q}^{\cN} \enspace. 
\end{align*}
It follows from Theorem~\ref{theorem:minimalifreverseatomic} that \(\cG{r}(\cN)\) is minimal.
\end{proof}

We conclude this section by collecting all the conditions described so far that guarantee that determinizing an automaton yields the minimal DFA.
\begin{corollary}\label{corol:detProducesMinimal}
Let \(\cN=(Q,\Sigma,\delta,I,F)\) be an NFA with \(L=\lang{\cN}\).
The following are equivalent:
\begin{alphaenumerate}%
\item \(\cG{r}(\cN)\) is minimal.
\item \(\mathord{\rrN} = \mathord{\rrL}\).
\item \(\forall u,v \in \Sigma^*, \; W_{\post_{u}^{\cN}(I),F}^{\cN} = W_{\post_{v}^{\cN}(I),F}^{\cN} \Lra \post_{u}^{\cN}(I) = \post_{v}^{\cN}(I)\).
\item \(\forall q \in Q,\; P_{\rrL}(W_{I,q}^{\cN}) = W_{I,q}^{\cN} \). \label{lemma:atomic-minimal}
\item \(\cN^R\) is atomic.
\end{alphaenumerate}
\end{corollary}

\subsection{Moore's Algorithm}

Given a DFA \(\cD\), Moore~\cite{moore1956} builds the minimal DFA for the language \(L=\lang{\cD}\) by removing unreachable states from \(\cD\) and then performing a stepwise refinement of an initial partition of the set of reachable states of \(\cD\).
Since we are interested in the refinement step, in what follows we assume that all DFAs have no unreachable states. 
In this section, we will describe Moore's state-partition \(\cQ^{\cD}\) and the right-language-based partition \(P_{\rrL}\) as greatest fixpoint computations and show that there exists an isomorphism between the two at each step of the fixpoint computation.
In fact, this isomorphism shows that Moore's DFA \(M\) satisfies \(P_{\rrL}(W_{I,q}^{M}) = W_{I,q}^{M}\) for every state \(q\).
Thus, by Theorem~\ref{theorem:minimalifreverseatomic}, \(M\) is isomorphic to \(\cF{r}(\lang{\cD})\).

First, we give Moore's algorithm which computes the state-partition that is later used to define Moore's  DFA.

\RemoveAlgoNumber
\begin{algorithm}
\SetAlgorithmName{Moore's Algorithm}{}
\SetSideCommentRight
\caption{Algorithm for constructing Moore's partition.}\label{alg:Moore}

\KwData{DFA \(\cD=\tuple{Q,\Sigma,\delta,I,F}\) with \(L = \lang{\cD}\).}
\KwResult{\(\cQ^{\cD} \in Part(Q)\).}

\(\cQ^{\cD} := \{F,F^c\}\),
\(\cQ' := \varnothing\)\;
\While{\(\cQ^{\cD}\neq \cQ'\)} {
	\(\cQ' := \cQ^{\cD}\)\;
	\ForAll{\(a\in \Sigma\)}{
		\(\cQ_a := \bigcurlywedge_{p \in \cQ^{\cD}} \{\pre_a^{\cD}(p), (\pre_a^{\cD}(p))^c\}\)\;
	}
	\(\cQ^{\cD} := \cQ^{\cD} \rc \bigcurlywedge_{a \in \Sigma} \cQ_a\)\;
}
\Return \(\cQ^{\cD}\)\;
\end{algorithm}
\vspace{-5pt}
\begin{definition}[Moore's DFA]
\label{def:MooreDFA}
Let \(\cD = (Q,\Sigma,\delta,I,F)\) be a DFA, and let \(\cQ^{\cD}\) be the partition of \(Q\) built by using Moore's algorithm.
\emph{Moore's DFA} for \(\lang{\cD}\) is \(M = (Q^M,\Sigma,\delta^M,I^M,F^M)\) where \(Q^M = \cQ^{\cD}\), \(I^M = \{\cQ^{\cD}(q) \mid q \in I\}\), \(F^M = \{\cQ^{\cD}(q) \mid q \in F\}\) and, for each \( S,S' \in Q^M\) and \(a \in \Sigma\), we have that \(\delta^M(S,a) = S'\) if{}f \(\exists q \in S, q' \in S'\) with \(\delta(q,a) = q'\).
\end{definition}

Next, we describe Moore's state-partition \(\cQ^{\cD}\) and the right-language-based partition \(P_{\rrL}\) as greatest fixpoint computations and show that there exists an isomorphism between the two at each step of the fixpoint computation.

\begin{definition}[Moore's state-partition]
\label{def:state-partition}
Let \(\cD=(Q,\Sigma,\delta,I,F)\) be a DFA.
Define  \emph{Moore's state-partition w.r.t.} \(\cD\), denoted by \(\cQ^{\cD}\), as follows. 
\[\cQ^{\cD} \ud \gfp(\lambda X. \bigcurlywedge\limits_{\substack{a \in \Sigma, S \in X}} \{\pre_a(S), (\pre_a(S))^c\} \rc \{F, F^c\}) \enspace .\]
\end{definition}

On the other hand, by Theorem~\ref{theoremF}~(\ref{theorem:minimalDFAlanguage}), each state of the minimal DFA for \(L\) corresponds to an equivalence class of \(\rrL\).
These equivalence classes can be defined in terms of non-empty intersections of complemented or uncomplemented right quotients of \(L\).

\begin{restatable}{lemmaR}{lemmacoatoms}\label{lemma:coatoms}
Let \(L\) be a regular language.
Then, for every \(u \in \Sigma^*\),
\[P_{\rrL}(u) = \bigcap\limits_{\substack{u \in Lw^{-1}\\w \in \Sigma^* }} Lw^{-1} \; \cap \bigcap\limits_{\substack{u \notin Lw^{-1}\\w \in \Sigma^* }} (Lw^{-1})^c\enspace .\]
\end{restatable}

It follows from Lemma~\ref{lemma:coatoms} that \(P_{\rrL} = \bigcurlywedge_{w \in \Sigma^*}\{Lw^{-1},(Lw^{-1})^c\}\), for every regular language \(L\) .
Thus, \(P_{\rrL}\) can also be obtained as a greatest fixpoint computation as follows.

\begin{restatable}{lemmaR}{lemmafixpointcoatoms}\label{lemma:fixpointcoatoms}
Let \(L\) be a regular language.
Then 
\begin{equation}\label{eq:fixpointcoatoms}
P_{\rrL} = \gfp(\lambda X. {\bigcurlywedge\limits_{\substack{a \in \Sigma, B \in X}}} \{Ba^{-1}, (Ba^{-1})^c\} \rc \{L, L^c\}) \enspace .
\end{equation}
\end{restatable}

The following result shows that, given a DFA \(\cD\) with \(L=\lang{\cD}\), there exists a partition isomorphism between  \(\cQ^{\cD}\) and \(P_{\rrL}\) at each step of the fixpoint computations given in  Definition~\ref{def:state-partition} and Lemma~\ref{lemma:fixpointcoatoms} respectively.

\begin{theorem}\label{theorem:postisomorphism}
Let \(\cD=(Q,\Sigma,\delta,I,F)\) be a DFA with \(L = \lang{\cD}\) and let \(\varphi:  \wp(Q) \ra \wp(\Sigma^*)\) be a function defined by \(\varphi(S) \ud W_{I, S}^{\cD}\).
Let \(\cQ^{\cD(n)}\) and \(P_{\rrL}^{(n)}\) be the \(n\)-th step of the fixpoint computation of  \(\cQ^{\cD}\) (Definition~\ref{def:state-partition}) and \(P_{\rrL}\) (Lemma~\ref{lemma:fixpointcoatoms}), respectively.
Then, \(\varphi\) is an isomorphism between \(\cQ^{\cD(n)}\) and \(P_{\rrL}^{(n)}\) for each \(n \geq 0\).
\end{theorem}

\begin{proof}
In order to show that \(\varphi\) is a partition isomorphism, it suffices to prove that \(\varphi\) is a bijective mapping between the partitions.
We first show that \(\varphi(\cQ^{\cD(n)}) = P_{\rrL}^{(n)}\), for every \(n \geq 0\).
Thus, the mapping \(\varphi\) is surjective.
Secondly, we show that \(\varphi\) is an injective mapping from  \(\cQ^{\cD(n)}\) to \(P_{\rrL}^{(n)}\).
Therefore, we conclude that \(\varphi\) is a bijection.

To show that \(\varphi(\cQ^{\cD(n)}) = P_{\rrL}^{(n)}\), for each \(n \geq 0\), we proceed by induction. 

\begin{itemize}
\item \emph{Base case:}
By definition, \(\cQ^{\cD(0)} = \{F,F^c\}\) and \(P_{\rrL}^{(0)} = \{L, L^c\}\).
Since \(\cD\) is deterministic (and complete), it follows that \(\varphi(F) = W_{I,F}^{\cD} = L\) and \(\varphi(F^c) = W_{I,F^c}^{\cD} = L^c\).

\item \emph{Inductive step:}
Before proceeding with the inductive step, we show that the following equations hold for each \(a,b \in \Sigma\) and \(S, S_i, S_j \in \cQ^{\cD(n)}\) with \(n \geq 0\):
\begin{align}
\varphi(\pre_a(S)^c) =  ((W_{I,S}^{\cD})a^{-1})^c & \label{eq:varphicomplement} \\
\varphi(\pre_a(S_i) \cap \pre_b(S_j)) = (W_{I,S_i}^{\cD})a^{-1} \cap (W_{I,S_j}^{\cD})b^{-1}& \label{eq:varphicap} \enspace .
\end{align}

For each \(S \in \cQ^{\cD(n)} \) and \(a \in \Sigma\) we have that:
\begin{align*}
\varphi(\pre_a(S)^c) & = \quad \text{[By definition of \(\varphi\)]} \\
W_{I,\pre_a(S)^c}^{\cD} & = \quad \text{[\(I = \{q_0\}\) and def. of \(W_{I,\pre_a(S)^c}^{\cD}\)]} \\
\{w \in \Sigma^* \mid \exists q \in \pre_a(S)^c,\; q = \hat{\delta}(q_0, w)\} & = \quad \text{[\(\cD\) is deterministic and complete]} \\
\{w \in \Sigma^* \mid \exists q \in \pre_a(S), \; q = \hat{\delta}(q_0, w)\}^c & = \quad 
\text{[By definition of \(\pre_a(S)\)]} \\
\{w \in \Sigma^* \mid \exists q \in S, \; q = \hat{\delta}(q_0, wa)\}^c & = \quad 
\text{[By definition of \((W_{I,S}^{\cD})a^{-1}\)]} \\
((W_{I, S}^{\cD})a^{-1})^c \enspace .
\end{align*}

Therefore Equation~\eqref{eq:varphicomplement} holds at each step of the fixpoint computation.
Consider now Equation~\eqref{eq:varphicap}. 
Let \(S_i, S_j \in \cQ^{\cD(n)}\).
Then,
\begin{align*}
\varphi(\pre_{a}(S_i) \cap \pre_b(S_j)) & = \quad \text{[By Def. \(\varphi\)]} \\
W_{I,(\pre_{a}(S_i) \cap \pre_b(S_j))}^{\cD} & = \quad \text{[\(I = \{q_0\}\) and def. \(W_{I,S}\)]} \\
\{w \in \Sigma^* \mid \exists q \in \pre_{a}(S_i) \cap \pre_b(S_j), q = \hat{\delta}(q_0,w)\} & = \quad \text{[By Def. of \(\cap \)]}\\ 
\{w \in \Sigma^* \mid \exists q \in \pre_{a}(S_i), \; q \in \pre_b(S_j), q = \hat{\delta}(q_0,w)\} & = \quad \text{[\(\cD\) is deterministic]}\\
W_{I, \pre_a(S_i)}^{\cD} \cap W_{I, \pre_b(S_j)}^{\cD} & = \quad \text{[By Def. of \((W_{I,S}^{\cD})a^{-1}\)]} \\
(W_{I, S_i}^{\cD})a^{-1} \cap (W_{I, S_j})b^{-1} & \enspace .
\end{align*}

Therefore Equation~\eqref{eq:varphicap} holds at each step of the fixpoint computation.

Let us assume that \(\varphi\left(\cQ^{\cD(n)}\right) = P_{\rrL}^{(n)}\) for every \(n \leq k\) with \(k > 0\).
Then,
\begin{adjustwidth}{-0.6cm}{}
\begin{align*}
\varphi\bigl(\cQ^{\cD(k{+}1)}\bigr) & = \text{[By Def.~\ref{def:state-partition} with \(X = \cQ^{\cD(k)}\)]} \\
\varphi\bigl(\bigcurlywedge\limits_{\substack{a \in \Sigma, S \in X}} \{\pre_a(S), \pre_a(S)^c\} \rc \{F, F^c\}\bigr) &= \text{[By Eqs.~\eqref{eq:varphicomplement},~\eqref{eq:varphicap} and def. of \(\bigcurlywedge\)]} \\
\bigcurlywedge\limits_{\substack{a \in \Sigma \\ \varphi(S) \in \varphi(X)}}\hspace{-5pt} \{(W_{I, S}^{\cD})a^{-1}, {((W_{I, S}^{\cD})a^{-1})}^c\} \rc \{L, L^c\} &= \text{[By induction hypothesis, \(\varphi(X) = P_{\rrL}^{(k)}\)]}\\[-6pt]
\bigcurlywedge\limits_{\substack{a \in \Sigma, B \in X'}}\hspace{-5pt} \{Ba^{-1}, {(Ba^{-1})}^c\} \rc \{L, L^c\} &= \text{[By Lemma~\ref{lemma:fixpointcoatoms} with \(X' = P_{\rrL}^{(k)}\)]}  \\[-6pt]
P_{\rrL}^{(k{+}1)} \enspace .
\end{align*}
\end{adjustwidth}
\end{itemize}

Finally, since \(\cD\) is a DFA then, for each \(S_i, S_j \in \cQ^{\cD(n)}\)(\(n \geq 0\)) with \(S_i \neq S_j\) we have that \(W_{I, S_i}^{\cD} \neq W_{I, S_j}^{\cD}\), i.e., \(\varphi (S_i) \neq \varphi(S_j)\).
Therefore, \(\varphi\) is an injective mapping.
\end{proof}

\begin{corollary}
\label{cor:stepwise-invariant}
Let \(\cD\) be a DFA with \(L = \lang{\cD}\).
Let \(\cQ^{\cD(n)}\) and \(P_{\rrL}^{(n)}\) be the \(n\)-th step of the fixpoint computation of \(\cQ^{\cD}\) and \(P_{\rrL}\) respectively.
Then, for each \(n \geq 0\),
\[P_{\rrL}^{(n)}(W_{I,S}^{\cD}) = W_{I,S}^{\cD}\enspace , \text{ for each \(S \in \cQ^{\cD(n)}\)} \enspace .\]
\end{corollary}
It follows that Moore's DFA \(M\), whose set of states corresponds to the state-partition at the end of the execution of Moore's algorithm, satisfies that \(\forall q \in Q^M,\;  P_{\rrL}(W_{I,q}^{M}) = W_{I,q}^{M}\) with \(L = \lang{M}\).
By Theorem~\ref{theorem:minimalifreverseatomic}, we have that \(\cG{r}(M)( = M\), since \(M\) is a DFA) is minimal.

\begin{restatable}{theoremR}{theoremMoore}
\label{thm:Moore}
Let \(\cD \) be a DFA and \(M\) be Moore's DFA for \(\lang{\cD}\) as in Definition~\ref{def:MooreDFA}.
Then, \(M\) is isomorphic to \(\cF{r}(\lang{\cD})\).
\end{restatable}

Finally, recall that Hopcroft~\cite{hopcroft1971} defined a DFA minimization algorithm which offers better performance than Moore's.
The ideas used by Hopcroft can be adapted to our framework to devise a new algorithm from computing \(P_{\rrL}\).
However, by doing so, we could not derive a better explanation than the one provided by Berstel et al.~\cite{berstel2010minimization}.

\section{Related Work and Conclusions}\label{sec:relatedWork}

Brzozowski and Tamm~\cite{brzozowski2014theory} showed that every regular language defines a unique NFA, which they call \emph{átomaton}.
The átomaton is built upon the minimal DFA \(\cN^{DM}\) for the language, defining its states as non-empty intersections of complemented or uncomplemented right languages of \(\cN^{DM}\), i.e., the atoms of the language.
They also observed that the atoms correspond to intersections of complemented or uncomplemented left quotients of the language.
Then they proved that the átomaton is isomorphic to the reverse automaton of the minimal deterministic DFA for the reverse language. 

Intuitively, the construction of the átomaton based on the right languages of the minimal DFA corresponds to \(\cG{\ell}(\cN^{DM})\), while its construction based on left quotients of the language corresponds to \(\cF{\ell}(\lang{\cN})\).

\begin{corollary}
\label{cor:atomaton}
Let \(\cN^{DM}\) be the minimal DFA for a regular language \(L\). Then, 
\begin{alphaenumerate}
\item \(\cG{\ell}(\cN^{DM})\) is isomorphic to the átomaton of \(L\).\label{cor:atomatonG}
\item \(\cF{\ell}(L)\) is isomorphic to the átomaton of \(L\).\label{cor:atomatonF}
\end{alphaenumerate}
\end{corollary}

In the same paper, they also defined the notion of \emph{partial átomaton} which is built upon an NFA \(\cN\).
Each state of the partial atomaton is a non-empty intersection of complemented or uncomplemented right languages of \(\cN\), i.e., union of atoms of the language.
Intuitively, the construction of the partial átomaton corresponds to \(\cG{\ell}(\cN)\).

\begin{corollary}
\label{cor:partial-atomaton}
Let \(\cN\) be an NFA. Then, \(\cG{\ell}(\cN)\) is isomorphic to the partial átomaton of \(\cN\).\break
\end{corollary}
\vspace{-10pt}

Finally, they also presented a number of results~\cite[Theorem 3]{brzozowski2014theory} related to the átomaton \(\cA\) of a minimal DFA \(\cD\) with \(L = \lang{\cD}\):
\begin{enumerate}
\item \(\cA\) is isomorphic to \(\cD^{RDR}\).\label{enum:ADRDR}
\item \(\cA^R\) is the minimal DFA for \(L^R\) 
\item \(\cA^D\) is the minimal DFA for \(L\).
\item \(\cA\) is isomorphic to \(\cN^{RDMR}\) for every NFA \(\cN\) accepting \(L\).\label{enum:ANRDMR}
\end{enumerate}

All these relations can be inferred from Figure~\ref{fig:DiagramResults} which connects all the automata constructions described in this paper together with the constructions introduced by Brzozowski and Tamm.
For instance, property~\ref{enum:ADRDR} corresponds to the path starting at \(\cN^{DM}\) (the minimal DFA for \(\lang{\cN}\)), labeled with \(R - \cG{r} - R \), and ending in the átomaton of \(\lang{\cN}\).
On the other hand, property~\ref{enum:ANRDMR} corresponds to the path starting at \(\cN\), labeled with \(R - \cF{r} - R\) and ending in the átomaton of \(\lang{\cN}\).
Finally, the path starting at \(\cN\), labeled with \(R - \cG{r} - R \) and ending in the partial átomaton of \(\cN\) shows that the later is isomorphic to \(\cN^{RDR}\).

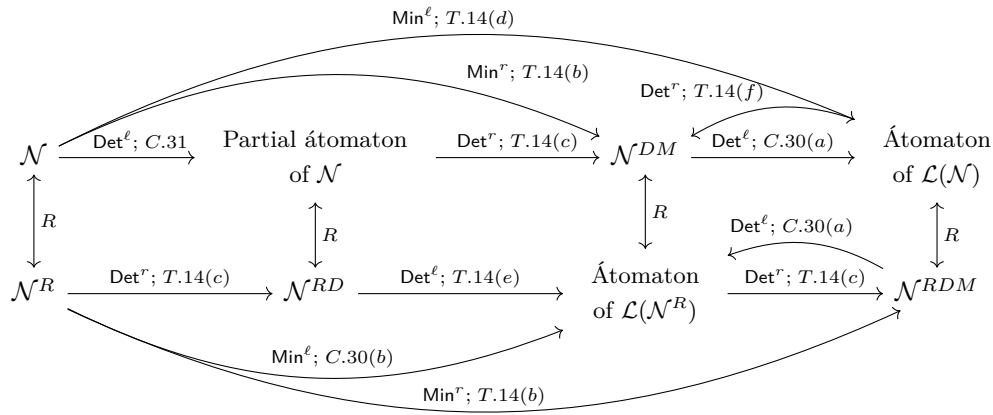
\begin{figure}[!ht]
\vspace{-7pt}
\centering
\begin{tikzcd}[column sep=huge, row sep=normal]
\cN \ar[r,"\cG{\ell};\; C.\ref{cor:partial-atomaton}", pos=0.6] \ar[d,"R",leftrightarrow] \ar[rr,"\cF{r}; \; T.\ref{theoremF}(\ref{theorem:minimalDFAlanguage})", bend left=25, near end] \ar[rrr,"\cF{\ell}; \;T.\ref{theoremF}(\ref{lemma:FlisomorphicRfrR})", bend left=25]& \begin{tabular}{c} Partial átomaton \\ of \(\cN\)\end{tabular} \ar[r,"\cG{r}; \; T.\ref{theoremF}(\ref{lemma:rightNDet})"] \ar[d,"R",leftrightarrow] & \cN^{DM} \ar[d,"R",leftrightarrow] \ar[r,"\cG{\ell};\;C.\ref{cor:atomaton}(\ref{cor:atomatonG})"] & \begin{tabular}{c} Átomaton \\ of \(\lang{\cN}\)\end{tabular} \ar[d,"R",leftrightarrow] \ar[l,"\cG{r}; \; T.\ref{theoremF}(\ref{lemma:LS+RS=RN})"',bend right=25, pos=0.5] \\
\cN^{R} \ar[r,"\cG{r}; \;T.\ref{theoremF}(\ref{lemma:rightNDet})"] \ar[rrr,"\cF{r}; \; T.\ref{theoremF}(\ref{theorem:minimalDFAlanguage})", bend right=25] \ar[rr,"\cF{\ell}; \;C.\ref{cor:atomaton}(\ref{cor:atomatonF})", bend right=24, pos=0.4] & \cN^{RD} \ar[r,"\cG{\ell}; \; T.\ref{theoremF}(\ref{lemma:AlRequalArNR})"] & \begin{tabular}{c} Átomaton \\ of \(\lang{\cN^{R}}\)\end{tabular} \ar[r,"\cG{r}; \; T.\ref{theoremF}(\ref{lemma:rightNDet})"] & \cN^{RDM} \ar[l,"\cG{\ell};\;C.\ref{cor:atomaton}(\ref{cor:atomatonG})"',bend right=25, pos=0.6]
\end{tikzcd}
\vspace{-15pt}
\caption{Extension of the diagram of Figure~\ref{Figure:diagramAutomata} including the átomaton and the partial átomaton.
Recall that \(\cN^{DM}\) is the minimal DFA for \(\lang{\cN}\).
The results referenced in the labels are those justifying the output of the operation.}\label{fig:DiagramResults}
\vspace{-7pt}
\end{figure}

In conclusion, we establish a connection between well-known independent minimization methods through Theorem~\ref{theorem:minimalifreverseatomic}.
Given a DFA, the left languages of its states form a partition on words, \(P\), and thus, each left language is identified by a state.
Intuitively, Moore's algorithm merges states to enforce the condition of Theorem~\ref{theorem:minimalifreverseatomic}, which results in merging blocks of \(P\) that belong to the same Nerode's equivalence class.
Note that Hopcroft's partition refinement method~\cite{hopcroft1971} achieves the same goal at the end of its execution though, stepwise, the partition computed may differ from Moore's.
On the other hand, any co-deterministic NFA satisfies the right-hand side of Equation~\eqref{eq:RightSetsDisjoint} hence, by Lemma~\ref{lemma:minimalifreverseatomic}, satisfies the condition of Theorem~\ref{theorem:minimalifreverseatomic}.
Therefore, the double-reversal method, which essentially determinizes a co-determinized NFA, yields the minimal DFA.
Finally, the left-right duality (Lemma~\ref{lemma:AutomataPartitionleftRightReverse}) of the language-based equivalences shows that the condition of Theorem~\ref{theorem:minimalifreverseatomic} is equivalent to that of Brzozowski and Tamm~\cite{brzozowski2014theory}.

Some of these connections have already been studied in order to offer a better understanding of Brzozowski's double-reversal method~\cite{Adamek2012,bonchi_algebra-coalgebra_2014,Champarnaud2002,Garcia2013}.
In particular, Adámek et al.~\cite{Adamek2012} and Bonchi et al.~\cite{bonchi_algebra-coalgebra_2014} offer an alternative view of minimization and determinization methods in a uniform way from a category-theoretical perspective.
In contrast, our work revisits these well-known minimization techniques relying on simple language-theoretical notions.

\appendix

\section{Deferred Proofs}
\CongruencePbwComplete*
\begin{proof}\hfill
\begin{enumerate}
\item \(\rr\) is a right congruence if{}f \(P_{\rr}(v)u \subseteq P_{\rr}(vu)\), for all \(u,v \in \Sigma^*\).

To simplify the notation, we denote \(P_{\rr}\), the partition induced by \(\rr\), simply by \(P\).\break
(\(\Ra\)).
Let \(x \in P(v)u\), i.e., \(x = \tilde{v}u\) with \(P(\tilde{v}) = P(v)\) (hence \(v \rr \tilde{v}\)). 
Since \(\rr\) is a right congruence and \(v \rr \tilde{v}\) then \(vu \rr \tilde{v}u\).
Therefore \(x \in P(vu)\).
\\(\(\La\)).
By hypothesis, for each \(u,v \in \Sigma^*\) and \(\tilde{v} \in P(v)\), \(\tilde{v}u \in P(vu)\).
Therefore, \(v \rr \tilde{v} \Ra \tilde{v}u \rr vu\).

\item\(\rl\) is a left congruence if{}f \(uP_{\rl}(v) \subseteq P_{\rl}(uv)\), for all \(u,v \in \Sigma^*\).

To simplify the notation, we denote \(P_{\rl}\), the partition induced by \(\rl\) simply by \(P\).\break
(\(\Ra\)).
Let \(x \in uP(v)\), i.e., \(x = u\tilde{v}\) with \(P(\tilde{v}) = P(v)\) (hence \(v \rl \tilde{v}\)). Since \(\rl\) is a left congruence and \(v \rl \tilde{v}\) then \(uv \rl u\tilde{v}\).
Therefore \(x \in P(uv)\).
\\(\(\La\)).
By hypothesis, for each \(u,v \in \Sigma^*\) and \(\tilde{v} \in P(v)\), \(u\tilde{v} \in P(uv)\), for all \(u \in \Sigma^*\).
Therefore \(v \rl \tilde{v} \Ra u\tilde{v} \rl uv\).\qedhere
\end{enumerate}
\end{proof}

\HrGeneratesL*

\begin{proof}
To simplify the notation, we denote \(P_{\rr}\), the partition induced by \(\rr\), simply by \(P\).
Let \(\cH = \cH^{r}(\rr, L) = (Q, \Sigma, \delta, I, F)\).
First, we prove that
\begin{equation}\label{eq:right-langs}
W_{I,P(u)}^{\cH} =   P(u), \quad \text{for each } u\in\Sigma^* \enspace .
\end{equation}

(\(\subseteq\)). We show that, for all \(w \in \Sigma^*\), \(w \in W_{I,P(u)}^{\cH} \Ra w \in P(u)\).
The proof goes by induction on length of \(w\).
\begin{itemize}
\item \emph{Base case:}
Let \(w = \varepsilon\) and \(\varepsilon \in W_{I,P(u)}^{\cH}\).
Note that the only initial state of \(\cH\) is \(P(\varepsilon)\).
Then, \(P(u) = \delta (P(\varepsilon),\varepsilon)\) and thus, \(P(u) = P(\varepsilon)\).
Hence, \(\varepsilon \in P(u)\).

Let \(w = a\) with \(a \in \Sigma\) and \(a \in W_{I,P(u)}^{\cH}\) .
Then, \(P(u) = \delta(P(\varepsilon), a)\).
By Definition~\ref{def:right-const}, \(P(\varepsilon)a \subseteq P(u)\).
Therefore, \(a \in P(u)\).

\item\emph{Inductive step:}
Now we assume by hypothesis of induction that, if \(|w| = n \;(n>1)\) then \(w \in W_{I,P(u)}^{\cH} \Ra w \in P(u)\).
Let \(|w| = n + 1\) and \(w \in W_{I,P(u)}^{\cH}\).
Assume w.l.o.g. that \(w = xa\) with \(x \in \Sigma^*\) and \(a \in \Sigma\).
Then, there exists a state \(q \in Q\) such that \(x \in W_{I,q}^{\cH}\) and \(P(u) = \delta(q, a)\).
Since \(x\) satisfies the induction hypothesis, we have that \(x \in q\), i.e., \(q\) denotes the state \(P(x)\).
On the other hand, by Definition~\ref{def:right-const}, we have that \(P(x)a \subseteq P(u)\).
Therefore, \(xa \in P(u)\).
\end{itemize}
(\(\supseteq\)). We show that, for all \(w \in \Sigma^*\), \(w \in P(u) \Ra w \in W_{I,P(u)}^{\cH}\).
Again, the proof goes by induction on length of \(w\).
\begin{itemize}
\item \emph{Base case:}
Let \(w = \varepsilon\) and \(\varepsilon \in P(u)\).
Then, \(P(u) = P(\varepsilon)\).
By Definition~\ref{def:right-const}, \(P(\varepsilon)\) is the initial state of \(\cH\).
Then, \( \varepsilon \in W_{I,P(\varepsilon)}^{\cH}\).

Let \(w = a\) with \(a \in \Sigma\) and \(a \in P(u)\).
Then \(P(u) = P(a)\).
Since \(P\) is a partition induced by a right congruence, by Lemma \ref{lemma:CongruencePbwComplete}, we have that \(P(\varepsilon)a \subseteq P(a)\).
Therefore, by Definition~\ref{def:right-const}, \(P(a) = \delta (P(\varepsilon), a)\).
Since \(P(\varepsilon)\) is the initial state of \(\cH\), we have that \(a \in W_{I,P(a)}^{\cH} \), i.e., \(w \in W_{I,P(u)}^{\cH}\).

\item\emph{Inductive step:}
Now we assume by hypothesis of induction that, if \(|w| = n\; (n>1)\) then \(w \in P(u) \Ra w \in W_{I,P(u)}^{\cH}\).
Let \(|w| = n + 1\) and \(w \in P(u)\).
Assume w.l.o.g. that \(w = xa\) with \(x \in \Sigma^*\) and \(a \in \Sigma\).
Then \(P(xa) = P(u)\).
Since \(P\) is a partition induced by a right congruence, by Lemma \ref{lemma:CongruencePbwComplete}, we have that \(P(x)a \subseteq P(xa)\).
Since \(x\in P(x)\), by induction hypothesis, \(x \in  W_{I,P(x)}^{\cH}\).
On the other hand, by Definition~\ref{def:right-const}, \(P(xa) = \delta(P(x), a) \).
Hence \(xa \in W_{I,P(xa)}^{\cH}\), i.e., \(w \in W_{I,P(u)}^{\cH}\).
\end{itemize}

We conclude this proof by showing that \(\lang{\cH} = L\).
\begin{align*}
\lang{\cH} &= \quad \text{[By definition of \(\lang{\cH}\)]} \\
\bigcup\limits_{q \in F} W_{I,q}^{\cH} \quad
& =  \quad \text{[By Definition~\ref{def:right-const}]}\\
\bigcup\limits_{\substack{P(w) \in Q \\ w \in L }} W_{I,P(w)}^{\cH} & = \quad  \text{[By Equation \eqref{eq:right-langs}]} \\
\bigcup\limits_{w \in L} P(w) & = \quad  \text{[By hypothesis, \(P(L) = L\)]} \\
L \enspace . \tag*{\qedhere}\\
\end{align*}
\end{proof}

\HlgeneratesL*

\begin{proof}
To simplify the notation, we denote \(P_{\rl}\), the partition induced by \(\rl\), simply by \(P\).
Let \(\cH = \cH^{\ell}(\rl, L) = (Q, \Sigma, \delta, I, F)\).
First, we prove that
\begin{equation}\label{eq:left-langs}
W_{P(u), F}^{\cH} =   P(u), \quad \text{for each } u\in\Sigma^* \enspace .
\end{equation}

(\(\subseteq\)). We show that, for all \(w \in \Sigma^*\), \(w \in W_{P(u), F}^{\cH} \Ra w \in P(u)\).
The proof goes by induction on length of \(w\).

\begin{itemize}
\item \emph{Base case:}
Let \(w = \varepsilon\) and \(\varepsilon \in W_{P(u), F}^{\cH}\).
Note that the only final state of \(\cH\) is \(P(\varepsilon)\).
Then, \(P(\varepsilon) \in \delta (P(u),\varepsilon)\) and thus, \(P(u) = P(\varepsilon)\).
Hence, \(\varepsilon \in P(u)\).

Let \(w = a\) with \(a \in \Sigma\) and \(a \in W_{P(u), F}^{\cH}\) .
Then, \(P(\varepsilon) \in \delta(P(u), a)\).
By Definition~\ref{def:left-const}, \(a P(\varepsilon) \subseteq P(u)\).
Therefore, \(a \in P(u)\).

\item\emph{Inductive step:}
Now we assume by hypothesis of induction that, if \(|w| = n\; (n>1)\) then \(w \in W_{P(u), F}^{\cH} \Ra w \in P(u)\).
Let \(|w| = n + 1\) and \(w \in W_{P(u), F}^{\cH}\).
Assume w.l.o.g. that \(w = ax\) with \(a \in \Sigma\) and \(x \in \Sigma^*\).
Then, there exists a state \(q \in Q\) such that \(x \in W_{q, F}^{\cH}\) and \(q \in \delta(P(u), a)\).
Since \(x\) satisfies the induction hypothesis, we have that \(x \in q\), i.e., \(q\) denotes the state \(P(x)\).
On the other hand, by Definition~\ref{def:left-const}, we have that \(aP(x) \subseteq P(u)\).
Therefore, \(ax \in P(u)\).
\end{itemize}

(\(\supseteq\)). We show that, for all \(w \in \Sigma^*\), \(w \in P(u) \Ra w \in W_{P(u), F}^{\cH}\).
Again, the proof goes by induction on length of \(w\).

\begin{itemize}
\item \emph{Base case:}
Let \(w = \varepsilon\) and \(\varepsilon \in P(u)\).
Then, \(P(u) = P(\varepsilon)\).
By Definition~\ref{def:right-const}, \(P(\varepsilon)\) is the final state of \(\cH\).
Then, \( \varepsilon \in W_{P(u), F}^{\cH}\).

Let \(w = a\) with \(a \in \Sigma\) and \(a \in P(u)\).
Then \(P(u) = P(a)\).
Since \(P\) is a partition induced by a left congruence, by Lemma \ref{lemma:CongruencePbwComplete}, we have that \(aP(\varepsilon) \subseteq P(a)\).
Therefore, by Definition~\ref{def:left-const}, \(P(\varepsilon) \in \delta (P(a), a)\).
Since \(P(\varepsilon)\) is the final state of \(\cH\), we have that \(a \in W_{P(a), F}^{\cH} \), i.e., \(w \in W_{P(u), F}^{\cH} \)~.

\item\emph{Inductive step:}
Now we assume by hypothesis of induction that, if \(|w| = n \; (n>1)\) then \(w \in P(u) \Ra w \in W_{P(u), F}^{\cH}\).
Let \(|w| = n + 1\) and \(w \in P(u)\).
Assume w.l.o.g. that \(w = ax\) with \(a \in \Sigma\) and \(x \in \Sigma^*\).
Then \(P(ax) = P(u)\).
Since \(P\) is a partition induced by a left congruence, by Lemma \ref{lemma:CongruencePbwComplete}, we have that \(aP(x) \subseteq P(ax)\).
Since \(x \in P(x)\), by induction hypothesis, \(x \in W_{P(x), F}^{\cH}\).
On the other hand, by Definition~\ref{def:left-const}, \(P(x) \in \delta(P(ax), a) \).
Hence \(ax \in W_{P(ax), F}^{\cH}\), i.e., \(w \in W_{P(u), F}^{\cH}\).
\end{itemize}

We conclude this proof by showing that \(\lang{\cH} = L\).
\begin{align*}
\lang{\cH} &= \quad \text{[By definition of \(\lang{\cH}\)]} \\
\bigcup\limits_{q \in I} W_{q, F}^{\cH} & = \quad \text{[By Definition~\ref{def:right-const}]}\\
\bigcup\limits_{\substack{P(w) \in Q \\ w \in L }} W_{P(w), F}^{\cH} & = \quad  \text{[By Equation \eqref{eq:left-langs}]} \\
\bigcup\limits_{w \in L} P(w) & = \quad  \text{[By hypothesis, \(P(L) = L\)]}\\
 L \enspace . \tag*{\qedhere}\\
\end{align*}
\end{proof}

\AutomataPartitionleftRightReverse*

\begin{proof}
Let \(\cH^{r}(\rr, L) = (Q, \Sigma, \delta, I, F)\) and \((\cH^{\ell}(\rl, L^R))^R = (\widetilde{Q}, \Sigma, \widetilde{\delta}, \widetilde{I}, \widetilde{F})\).
We will show that \(\cH^{r}(\rr, L)\) is isomorphic to \((\cH^{\ell}(\rl, L^R))^R\).

Let \(\varphi: Q \rightarrow \widetilde{Q}\) be a mapping assigning to each state \(P_{\rr}(u) \in Q\) with \(u \in \Sigma^*\), the state \(P_{\rl}(u^R) \in \widetilde{Q}\).
We show that \(\varphi\) is an NFA isomorphism between \(\cH^{r}(\rr, L)\) and \((\cH^{\ell}(\rl, L^R))^R\).

The initial state \(P_{\rr}(\varepsilon)\) of \(\cH^{r}(\rr, L)\) is mapped to \(P_{\rl}(\varepsilon)\) which is the final state of \(\cH^{\ell}(\rl, L^R)\), i.e., the initial state of \((\cH^{\ell}(\rl, L^R))^R\).

Each final state \(P_{\rr}(u)\) of \(\cH^{r}(\rr, L)\) with \(u \in L\) is mapped to \(P_{\rl}(u^R)\), where \(u^R \in L^R\).
Therefore, \(P_{\rl}(u^R)\) is an initial state of \(\cH^{\ell}(\rl, L^R)\), i.e., a final state of \((\cH^{\ell}(\rl, L^R))^R\).

Now, note that, by Definition~\ref{def:left-const}, \(\cH^{\ell}(\rl, L^R)\) is a co-DFA, therefore \((\cH^{\ell}(\rl, L^R))^R\) is a DFA.
Let us show that \(q' = \delta(q, a)\) if and only if \(\varphi(q') = \widetilde{\delta}(\varphi(q),a)\), for all \(q, q' \in Q\) and \(a \in \Sigma\).
Assume that \(q = P_{\rr}(u)\) for some \(u \in \Sigma^*\), and \(q' = \delta(q, a)\) with \(a \in \Sigma\).
By Definition~\ref{def:right-const}, we have that \(q' = P_{\rr}(ua)\).
Then, \(\varphi(q) = P_{\rl}(u^R)\) and \(\varphi(q') = P_{\rl}(au^R)\).
Since \(\rl\) is a left congruence, using Lemma \ref{lemma:CongruencePbwComplete} we have that \(aP_{\rl}(u^R) \subseteq P_{\rl}(au^R)\).
Then, there is a transition in \(\cH^{\ell}(\rl, L^R)\) from state \(\varphi(q') = P_{\rl}(au^R)\) to state \(\varphi(q) = P_{\rl}(u^R)\) reading \(a\).
Hence, there exists the reverse transition in \((\cH^{\ell}(\rl, L^R))^R\), i.e., \(\varphi(q') = \widetilde{\delta}(\varphi(q),a)\).

Assume now that  \(\widetilde{q} = P_{\rl}(u^R)\) for some \(u \in \Sigma^*\), and \(\widetilde{q'} = \widetilde{\delta}(\widetilde{q}, a)\) with \(a \in \Sigma\).
By Definition~\ref{def:left-const}, we have that \(\widetilde{q'} = P_{\rl}(au^R)\).
Consider a state \(q \in Q\) such that \(\varphi(q) = \widetilde{q}\), then \(q\) is of the form \(P_{\rr}(u)\).
Likewise, consider a state \(q' \in Q\) such that \(\varphi(q') = \widetilde{q'}\), then \(q'\) is of the form \(P_{\rr}(ua)\).
Since \(P_{\rr}\) is a partition induced by a right congruence, using Lemma \ref{lemma:CongruencePbwComplete}, we have that \(P_{\rr}(u)a \subseteq P_{\rr}(ua)\) and thus, \(q' = \delta (q,a)\).
\end{proof}

\languageEqualPost*
\begin{proof}
For each \(u,v \in \Sigma^*\),
\begin{align*}
u \rrL v \Lra u \rrN v & \iff \quad\text{[By~\eqref{eq:Rlanguage} and~\eqref{eq:RState}]} \\
u^{-1}L = v^{-1}L \Lra \post_u^{\cN}(I) = \post_v^{\cN}(I) & \iff \quad \text{[Definition of quotient of \(L\)]} \\ 
W_{\post_{u}^{\cN}(I),F}^{\cN} = W_{\post_{v}^{\cN}(I),F}^{\cN} \Lra \post_u^{\cN}(I) = \post_v^{\cN}(I) & \tag*{\qedhere} \enspace .
\end{align*}
\end{proof}

\theoremF*

\begin{proof}\hfill
\begin{enumerate}[(a)]
\item \(\lang{\cF{r}(L)} = \lang{\cF{\ell}(L)} = L = \lang{\cG{r}(\cN)} = \lang{\cG{\ell}(\cN)}\).%

By Definition \ref{def:FG}, \(\cF{r}(L) = \cH^{r}(\rrL, L)\) and \(\cG{r}(\cN) = \cH^{r}(\rrN, L)\).
By Lemma \ref{lemma: HrGeneratesL}, \(\lang{\cH^{r}(\rrL, L)} = L = \lang{\cH^{r}(\rrN, L)}\).
Therefore, \(\lang{\cF{r}(L)} = \cG{r}(\cN) = L\).
The proof of \(\lang{\cF{\ell}(L)} = L = \lang{\cG{\ell}(\cN)}\) goes similarly using Lemma~\ref{lemma:HlgeneratesL}.

\item \(\cF{r}(L)\) is isomorphic to the minimal deterministic automaton for \(L\).%

Let \(P\) be the partition induced by \(\rrL\).
Recall that the automaton \(\cF{r}(L) = (Q, \Sigma, \delta, I, F)\)  is a complete DFA (see Remark \ref{remark:deterministic}).
Recall also that the \emph{quotient DFA} of \(L\), defined as \(\cD = (\widetilde{Q}, \Sigma, \eta, \widetilde{q}_0, \widetilde{F})\) where \(\widetilde{Q} = \{u^{-1}L \mid u \in \Sigma^*\}\), \(\eta(u^{-1}L, a) = a^{-1}(u^{-1}L) \) for each \(a \in \Sigma\), \(\widetilde{q}_0 = \varepsilon^{-1}L = L\) and \(\widetilde{F} = \{u^{-1}L \mid \varepsilon \in u^{-1}L\}\), is the minimal DFA for \(L\).
We will show that \(\cF{r}(L)\) is isomorphic to \(\cD\).

Let \(\varphi: \widetilde{Q} \rightarrow Q\) be the mapping assigning to each state \(\widetilde{q}_i \in \widetilde{Q}\) of the form \(u^{-1}L\), the state \(P(u) \in Q\), with \(u \in \Sigma^*\).
Note that, in particular, if \(\widetilde{q_i} \in \widetilde{Q}\) is the empty set, then \(\varphi\) maps \(\widetilde{q_i}\) to the block in \(P\) that contains all the words that are not prefixes of \(L\). 
We show that \(\varphi\) is a DFA isomorphism between \(\cD\) and \(\cF{r}(L)\).

The initial state \(\widetilde{q}_0 = \varepsilon^{-1}L\) of \(\cD\) is mapped to the state \(P(\varepsilon)\) which, by definition, is the unique initial state of \(\cF{r}(L)\).
Each final state \(u^{-1}L \in \widetilde{F}\) is mapped to the state \(P(u)\) with \(u \in L\) which, by definition, is a final state of \(\cF{r}(L)\).

We now show that \(\widetilde{q}_j = \eta(\widetilde{q}_i, a)\) if and only if \(\varphi(\widetilde{q}_j) = \delta(\varphi(\widetilde{q}_i),a)\), for all \(\widetilde{q}_i, \widetilde{q}_j \in \widetilde{Q}, a \in \Sigma\).
Assume that \(\widetilde{q}_i = u^{-1}L\) for some \(u \in \Sigma^*\) and \(\widetilde{q}_j = \eta(\widetilde{q}_i, a)\) where \(\widetilde{q}_j = a^{-1}(u^{-1}L)\) and \(a \in \Sigma\).
Note that \(a^{-1}(u^{-1}L) = \{x \in \Sigma^* \mid uax \in L\}\).
Then, \(\varphi(\widetilde{q}_i) = P(u)\) and \(\varphi(\widetilde{q}_j) = P(ua)\).
Since \(P\) is a partition induced by a right congruence, using Lemma~\ref{lemma:CongruencePbwComplete}, we have that \(P(u)a \subseteq P(ua)\).
Therefore, \(\varphi(\widetilde{q}_j)= \delta(\varphi(\widetilde{q}_i),a)\).

Assume now that \(P(ua) = \delta(P(u), a)\)  for some \(u \in \Sigma^*\) and \(a \in \Sigma\).
Consider \(\widetilde{q}_i \in \widetilde{Q}\) such that \(\varphi(\widetilde{q}_i) = P(u)\), then \(\widetilde{q}_i = u^{-1}L\).
Likewise,  consider \(\widetilde{q}_j \in \widetilde{Q}\) such that \(\varphi(\widetilde{q}_j)= P(ua)\), then \(\widetilde{q}_j=(ua)^{-1}L = a^{-1}(u^{-1})L\).
Therefore, \(\widetilde{q}_j = \eta (\widetilde{q}_i, a)\).

\item \(\cG{r}(\cN)\) is isomorphic to \(\cN^{D}\).

Recall that, given \(\cN = (Q, \Sigma, \delta, I, F)\), \(\cN^D\) denotes the DFA that results from applying the subset construction to \(\cN\) and removing all states that are not reachable.
Thus \(\cN^D\) possibly contains empty states but no state is unreachable.
Let \(\cN^D = (Q_{d}, \Sigma, \delta_d, \{I\}, F_d)\) and let \(\cG{r}(\cN) = (\widetilde{Q}, \Sigma, \widetilde{\delta}, \widetilde{I}, \widetilde{F})\).
Let \(P\) be the partition induced by \(\rrN\) and let \(\varphi: \widetilde{Q} \rightarrow Q_{d}\) be the mapping assigning to each state \(P(u) \in \widetilde{Q}\), the set \(\post_u^{\cN}(I) \in Q_{d}\) with \(u \in \Sigma^*\).
Note that if \(u \in \Sigma^*\) is not a prefix of \(\lang{\cN}\), then \(\varphi\) maps \(P(u)\) to \(\post_{u}^{\cN}(I) = \emptyset\).
We show that \(\varphi\) is a DFA isomorphism between \(\cG{r}(\cN)\) and \(\cN^{D}\).

The initial state of \(\cG{r}(\cN)\), \(P(\varepsilon)\), is mapped to \(\post_\varepsilon^{\cN}(I)= \{I\}\).
Therefore, \(\varphi\) maps the initial state of \(\cG{r}(\cN)\) to the initial state of \(\cN^{D}\).
Each final state of \(\cG{r}(\cN)\), \(P(u)\) with \(u \in L\), is mapped to \(\post_u^{\cN}(I)\).
Since \(\post_u^{\cN}(I) \cap F \neq \emptyset\), \(\post_u^{\cN}(I) \in \widetilde{F}\).

Now note that, by Remark~\ref{remark:deterministic}, \(\cG{r}(\cN)\) is a complete DFA, and by construction, so is \(\cN^D\).
Let us show that \(\widetilde{q'} = \widetilde{\delta}(\widetilde{q},a)\) if{}f \(\varphi(\widetilde{q'}) = \delta_d(\varphi(\widetilde{q}),a)\), for all \(\widetilde{q}, \widetilde{q'} \in \widetilde{Q}\) and \(a \in \Sigma\).
Assume that \(\widetilde{q} = P(u)\), for some \(u \in \Sigma^*\), and \(\widetilde{q'} = \widetilde{\delta}(\widetilde{q},a)\), with \(a \in \Sigma\).
By Definition~\ref{def:right-const}, we have that \(\widetilde{q'} = P(ua)\).
Then, \(\varphi(\widetilde{q}) = \post_u^{\cN}(I)\) and \(\varphi(\widetilde{q'}) = \post_{ua}^{\cN}(I) = \post_a^{\cN}(\post_u^{\cN}(I))\).
Therefore, \(\varphi(\widetilde{q'})  = \delta_d(\varphi(\widetilde{q'}), a)\).

Assume now that \(\delta_d(\post_u^{\cN}(I),a) =  \post_{ua}^{\cN}(I)\).
Consider \(\widetilde{q} \in \widetilde{Q}\) such that \(\varphi(\widetilde{q}) = \post_u^{\cN}(I)\), then \(\widetilde{q} = P(u)\).
Likewise, consider \(\widetilde{q'} \in \widetilde{Q}\) such that \(\varphi(\widetilde{q'}) = \post_{ua}^{\cN}(I)\), then \(\widetilde{q'} = P(ua)\).
Since \(P\) is a partition induced by a right congruence, using Lemma~\ref{lemma:CongruencePbwComplete}, we have that \(P(u)a \subseteq P(ua)\).
Therefore, \(\widetilde{q'} = \widetilde{\delta}(\widetilde{q},a)\).

\item \(\cF{\ell}(L)\) is isomorphic to \((\cF{r}(L^R))^R\)%

Observe that, for each \(u \in \Sigma^*\):
\begin{multline}\label{eq:um1L} 
(u^{-1}L)^R = \{x^R \in \Sigma^* \mid ux \in L\} = \{x^R \in \Sigma^* \mid x^R u^R \in L^R\} =\\ \{x' \in \Sigma^* \mid x' u^R \in L^R\} = L^R(u^R)^{-1} \enspace .
\end{multline}
Therefore,
\begin{align*}
u \rlL v & \Lra \quad \text{[By Definition ~\eqref{eq:Llanguage}]} \\
u^{-1}L = v^{-1}L & \Lra \quad \text{[\(x = y \Lra x^R = y^R\)]}\\
(u^{-1}L)^R = (v^{-1}L)^R & \Lra \quad\text{[By Equation ~\eqref{eq:um1L}]} \\
L^R(u^R)^{-1} = L^R(v^R)^{-1} & \Lra\quad\text{[By Definition ~\eqref{eq:Rlanguage}]} \\
u^R \rr_{L^R} v^R \enspace . & 
\end{align*}
Finally, it follows from Lemma~\ref{lemma:AutomataPartitionleftRightReverse} that \(\cF{\ell}(L)\) is isomorphic to \((\cF{r}(L^R))^R\).

\item \(\cG{\ell}(\cN)\) is isomorphic to \((\cG{r}(\cN^R))^R\).

For each \(u,v \in \Sigma^*\):
\begin{align*}
u \sim_{\cN^R}^{\ell} v & \Lra \quad\text{[By Defintion~\ref{def:automataEquiv}]}\\
\pre_u^{\cN^R}(F) = \pre_v^{\cN^R}(F) & \Lra\quad \text{[\(q \in \pre^{\cN^R}_{x}(F)\) if{}f \(q \in \post^{\cN}_{x^R}(I) \)]}\\
\post_{u^R}^{\cN}(I) = \post_{v^R}^{\cN}(I) & \Lra \quad\text{[By Definition~\ref{def:automataEquiv}]}\\
u^R \rlN v^R \enspace .
\end{align*}
It follows from Lemma~\ref{lemma:AutomataPartitionleftRightReverse} that \(\cG{\ell}(\cN)\) is isomorphic to \(\cG{r}(\cN^R))^R\).

\item \(\cG{r}(\cG{\ell}(\cN))\) is isomorphic to \(\cF{r}(L)\).

By Theorem~\ref{theoremF}~(\ref{lemma:language-F}), \(\cG{\ell}(\cN)\) is a co-deterministic automaton generating the language \(\lang{\cN}\).
Since \(\cG{\ell}(\cN)\) is co-deterministic, it satisfies Equation~\eqref{eq:RightSetsDisjoint} from Theorem~\ref{lemma:languageEqualPost}.
Therefore, \(\cG{r}(\cG{\ell}(\cN))\) is isomorphic to \(\cF{r}(\lang{\cG{\ell}(\cN)})=\cF{r}(\lang{\cN})\).
\qedhere
\end{enumerate}
\end{proof}

\atoms*
\begin{proof}
For each \(u \in \Sigma^*\), define \(L_u = \bigcap\limits_{\substack{u \in w^{-1}L \\w \in \Sigma^* }} w^{-1}L \bigcap\limits_{\substack{u \notin w^{-1}L \\w \in \Sigma^* }} (w^{-1}L)^c\).
First, we show that \(P_{\rlL}(u) \subseteq L_u\), for each \(u \in \Sigma^*\).
Let \(v \in P_{\rlL}(u)\), i.e., \(Lu^{-1} = Lv^{-1}\).
Then, for each \(w \in \Sigma^*\), \(u \in w^{-1}L \Lra wu \in L \Lra w \in Lu^{-1} \Lra w \in Lv^{-1} \Lra v \in w^{-1}L\).
Therefore, \(\forall v \in P_{\rlL}(u), \; v \in L_u\) and thus, \(P_{\rlL}(u) \subseteq L_u\).

Next, we show that \(L_u \subseteq P_{\rlL}(u)\).
Let \(v \in L_u\).
Then,  \(\forall w \in \Sigma^*\), \(u \in w^{-1}L \Lra v \in w^{-1}L\).
It follows that \(w \in Lu^{-1} \Lra w \in Lv^{-1}\) and, therefore, \(v \in P_{\rlL}(u)\).
\end{proof}

\lemmacoatoms*
\begin{proof}
For each \(u \in \Sigma^*\), define \(L_u = \bigcap\limits_{\substack{u \in Lw^{-1} \\w \in \Sigma^* }} Lw^{-1} \bigcap\limits_{\substack{u \notin Lw^{-1} \\w \in \Sigma^* }} (Lw^{-1})^c\).
First, we show that \(P_{\rrL}(u) \subseteq L_u\), for each \(u \in \Sigma^*\).
Let \(v \in P_{\rrL}(u)\), i.e., \(u^{-1}L = v^{-1}L\).
Then, for each \(w \in \Sigma^*\), \(u \in Lw^{-1} \Lra uw \in L \Lra w \in u^{-1}L \Lra w \in v^{-1}L \Lra v \in Lw^{-1}\).
Therefore, \(\forall v \in P_{\rrL}(u), \; v \in L_u\) and thus, \(P_{\rrL}(u) \subseteq L_u\).

Next, we show that \(L_u \subseteq P_{\rrL}(u)\).
Let \(v \in L_u\).
Then,  \(\forall w \in \Sigma^*\), \(u \in Lw^{-1} \Lra v \in Lw^{-1}\).
It follows that \(w \in u^{-1}L \Lra w \in v^{-1}L\) and, therefore, \(v \in P_{\rrL}(u)\).
\end{proof}

\lemmafixpointcoatoms*
\begin{proof}
Let \(\Sigma^{\leq n}\) (resp.  \(\Sigma^{n}\)) denote the set of words with length up to \(n\) (resp. exactly \(n\)), i.e., \(\Sigma^{\leq n }\ud \{w \in \Sigma^* \mid |w| \leq n\}\) (resp. \(\Sigma^{ n }\ud \{w \in \Sigma^* \mid |w| = n\}\)).
Let us denote \(X^{n}\), the \(n\)-th iteration of the greatest fixpoint computation of Equation~\eqref{eq:fixpointcoatoms}.
 We will prove by induction on \(n\) that the following equation holds for each \(n \geq 0\):
\begin{equation}\label{eq:fixpoint-induction}
X^{n{+}1} = {\bigcurlywedge\limits_{\substack{a \in \Sigma, B \in X^{n}}}} \{Ba^{-1}, (Ba^{-1})^c\} \rc \{L, L^c\} = \bigcurlywedge_{w \in \Sigma^{\leq n}}\{Lw^{-1},(Lw^{-1})^c\} \enspace .
\end{equation}
\pagebreak
\begin{itemize}
\item \emph{Base case:}
Let \(n = 0\). It is easy to see that the Equation~\eqref{eq:fixpoint-induction} holds since \(\{L, L^c\} = \{L\varepsilon^{-1},(L\varepsilon^{-1})^c\}\).
Now, let \(n =1\). Then,
\begin{align*}
{\bigcurlywedge\limits_{\substack{a \in \Sigma, B \in X^{1}}}} \{Ba^{-1}, (Ba^{-1})^c\} \rc \{L, L^c\} &=\quad[X^{1} = \{L, L^c\}]\\
\bigcurlywedge_{a \in \Sigma}\left(\{La^{-1}, (La^{-1})^c\} \rc \{(L^c)a^{-1}, ((L^c)a^{-1})^c\}\right) \rc \{L, L^c\} &=\quad[(La^{-1})^c = L^ca^{-1}] \\
\bigcurlywedge_{a \in \Sigma}\{La^{-1}, (La^{-1})^c\} \rc \{L, L^c\} &=\quad[\Sigma^{\leq 1} = \{\varepsilon\} \cup \Sigma] \\
\bigcurlywedge_{a \in \Sigma, w \in \Sigma^{\leq 1}}\{Lw^{-1},(Lw^{-1})^c\} \enspace . &
\end{align*}
\item \emph{Inductive Step:}
Let us assume that Equation~\eqref{eq:fixpoint-induction} holds for each \(n \leq k\).
We will prove that it holds for \(n = k+1\).
Note that, using the inductive hypothesis twice, we have that: 
\begin{align}
X^{k+1} = \bigcurlywedge_{w \in \Sigma^{\leq k}}\{Lw^{-1},(Lw^{-1})^c\} &= \nonumber\\
\bigcurlywedge_{w \in \Sigma^{\leq k-1}}\{Lw^{-1},(Lw^{-1})^c\} \curlywedge \bigcurlywedge_{a \in \Sigma, w \in \Sigma^{k-1}} \{Lw^{-1}a^{-1},(Lw^{-1}a^{-1})^c\} &= \nonumber\\
X^{k} \curlywedge \bigcurlywedge_{w \in \Sigma^{k}} \{Lw^{-1},(Lw^{-1})^c\} \label{eq:Xk+1}\enspace .
\end{align}
Using Equation~\eqref{eq:Xk+1}, the identities \((La^{-1})^c = L^ca^{-1}\) and \(Ba^{-1}\cap \widetilde{B}a^{-1} = (B \cap \widetilde{B})a^{-1}\) and the induction hypothesis, it follows that:
\begin{adjustwidth}{-0.5cm}{}
\begin{align*}
X^{k+2} =
{\bigcurlywedge_{a \in \Sigma, B \in X^{k+1}}} \{Ba^{-1}, (Ba^{-1})^c\} \rc \{L, L^c\} &=\\
{\bigcurlywedge_{a \in \Sigma, B \in X^{k}}} \{Ba^{-1}, (Ba^{-1})^c\} \rc {\bigcurlywedge_{a \in \Sigma, B \in \bigcurlywedge_{w \in \Sigma^{k}} \{Lw^{-1},(Lw^{-1})^c\}}} \{Ba^{-1}, (Ba^{-1})^c\}\rc \{L, L^c\} &= \\
\bigcurlywedge_{w \in \Sigma^{\leq k}}\{Lw^{-1},(Lw^{-1})^c\} \rc \bigcurlywedge_{w \in \Sigma^{k+1}}\{Lw^{-1},(Lw^{-1})^c\} \rc \{L, L^c\} &=\\
\bigcurlywedge_{w \in \Sigma^{\leq k+1}}\{Lw^{-1},(Lw^{-1})^c\} \enspace . &
\end{align*}
\end{adjustwidth}
We conclude that \(P_{\rrL} = \gfp(\lambda X. {\bigcurlywedge\limits_{\substack{a \in \Sigma, B \in X}}} \{Ba^{-1}, (Ba^{-1})^c\} \rc \{L, L^c\})\).\qedhere
\end{itemize}
\end{proof}

\theoremMoore*
\begin{proof}
Let \(\cD = (Q', \Sigma, \delta', I', F')\).
Recall that Moore's minimal DFA is defined as \(M = (Q, \Sigma, \delta, I, F)\) where the set of states corresponds to Moore's state-partition w.r.t. \(\cD\), i.e., \(Q = \cQ^{\cD}\); \(I = \{\cQ^{\cD}(q) \mid q \in I'\}\); \(F = \{\cQ^{\cD}(q) \mid q \in F'\}\) and \(S' = \delta(S,a)\) if{}f \(\exists q \in S, q' \in S': q' = \delta'(q,a)\), for each \(S, S' \in Q\) and \(a \in \Sigma\).
Let \(\cF{r}(\lang{\cD}) = (\widetilde{Q}, \Sigma, \widetilde{\delta}, \widetilde{I}, \widetilde{F})\) be described as in Definition \ref{def:FG}.
Finally, let \(L\) denote \(\lang{\cD}\), for simplicity.
By Theorem~\ref{theorem:postisomorphism}, the mapping \(\varphi:  \wp(Q') \ra \wp(\Sigma^*)\) defined as \(\varphi(S) = W_{I',S}^{\cD}\), for each \(S \in \cQ^{\cD}\), is a partition isomorphism between \(\cQ^{\cD}\) and \(P_{\rrL}\).
Note that, by construction of \(M\), \(W_{I,S}^{M} = W_{I',S}^{\cD}\), for each \(S \in \cQ^{\cD}\).
Thus, the mapping \(\psi: Q \ra \widetilde{Q}\) defined as \(\psi(S) = W_{I,S}^{M} \), for each \(S \in Q\), is also a partition isomorphism between \(\cQ^{\cD}\) and \(P_{\rrL}\).
In fact, we will show that \(\psi\) is a DFA morphism between \(M\) and \(\cF{r}(L)\).

The initial state \(I\) of \(M\) is mapped to \(\psi(I) = W_{I, I}^{M} = P(\varepsilon)\), since \(\varepsilon \in W_{I, I}^{M} \).
Therefore, \(\psi\) maps the initial state of \(M\) with the initial state of \(\cF{r}(L)\).
Note that each final state \(S\) in \(F\) is such that \(S \subseteq F' \).
Therefore, \(\psi(S) = W_{I,S}^{M} = P(u)\) with \(u \in L\), i.e., \(\psi\) maps each final state of \(M\) to a final state of \(\cF{r}(L)\).

We also have to show that \(S' = \delta(S,a)\) if{}f \(\psi(S') = \widetilde{\delta}(\psi(S),a)\), for all \(S,S' \in Q\) and \(a \in \Sigma\).
Assume that \(S' = \delta(S,a)\), for some \(S,S' \in Q\) and \(a \in \Sigma\).
Therefore, there exists \(q,q' \in Q'\) such that \(q \in S, q' \in S'\) and \(q' = \delta'(q,a)\).
Then, \(\psi(S) = W_{I,S}^{M}\) and \(\psi(S') = W_{I,S'}^{M}\) and there exists \(u \in W_{I,S}(M)\) such that \(ua \in W_{I,S'}^{M}\) (recall that \(M\) is a DFA and therefore complete).
Then, \(\psi(S) = P(u)\) and \(\psi(S') = P(ua)\).
Since \(P\) is a partition induced by a right congruence then, using Lemma~\ref{lemma:CongruencePbwComplete}, \(P(u)a \subseteq P(ua)\).
Therefore, \(\psi(S') = \widetilde{\delta}(\psi(S),a)\).
Assume now that, \(P(ua) = \widetilde{\delta}(P(u),a)\) for some \(u\in \Sigma^*\) and \(a\in \Sigma\).
Consider \(S \in Q\) such that \(\psi(S) = P(u)\), then \(u\) belongs to the left language of \(S\), i.e., \(u \in W_{I,S}^{M}\).
Likewise, consider \(S' \in Q\) such that \(\psi(S') = P(ua)\), then \(ua \in W_{I, S'}^{M}\).
Therefore, there exists \(q,q' \in Q'\) such that \(q \in S, q' \in S'\) and \(q' = \delta'(q,a)\).
Thus, \(S' = \delta(S,a)\).
\end{proof}
Finally we prove the next two results related to Definitions \ref{def:language} and \ref{def:automataEquiv} in Section \ref{sec:congruences}.

\begin{lemma}
\label{lemma:lang-cong}
Let \(L \subseteq \Sigma^*\) be a regular language.
Then, the following holds:
\begin{romanenumerate}%
\item \(\rrL\) is a right congruence; 
\item \(\rlL\) is a left congruence; and
\item \(P_{\rrL}(L) = L = P_{\rlL}(L)\).
\end{romanenumerate}
\end{lemma}

\begin{proof}
Let us prove that \(\rrL\) is a right congruence.
Assume \(u \rrL v\), i.e., \(u^{-1}L = v^{-1}L\).
Given \(x \in \Sigma^*\), we have that,
\[(ux)^{-1}L = x^{-1}(u^{-1}L) = x^{-1}(v^{-1}L) = (vx)^{-1}L \enspace .\]
Therefore, \(ux \rrL vx\).

Now, let us prove that \(\rlL\) is a left congruence.
Assume \(u \rlL v\), i.e., \(Lu^{-1} = Lv^{-1}\).
Given \(x \in \Sigma^*\), we have that,
\[L(xu)^{-1} = (Lu^{-1})x^{-1} =  (Lv^{-1})x^{-1} = L(xv)^{-1} \enspace .\]
Therefore, \(xu \rrL xv\).

Finally, let \(P_{\rrL}\) be the finite partition induced by \(\rrL\).
We show that \(P_{\rrL}(L) = L\).
First note that \(L \subseteq P_{\rrL}(L)\) by the reflexivity of the equivalence relation \(\rrL\).
On the other hand, we prove that for every \(u \in \Sigma^*\), if \(u \in P_{\rrL}(L)\) then \(u \in L\).
By hypothesis, there exists \(v \in L\) such that \(u \rrL v\), i.e., \(u^{-1}L = v^{-1}L\).
Since \(v \in L\) then \(\varepsilon \in v^{-1}L\).
Therefore, \(\varepsilon \in u^{-1}L\) and we conclude that \(u \in L\).

The proof of \(P_{\rlL}(L) = L\) goes similarly.
\end{proof}

\begin{lemma}
\label{lemma:automata-cong}
Let \(\cN\) be an NFA.
Then, the following holds:
\begin{romanenumerate}%
\item \(\rrN\) is a right congruence; 
\item \(\rlN\) is a left congruence; and
\item \(P_{\rrN}(\lang{\cN}) = \lang{\cN} = P_{\rlN}(\lang{\cN})\).
\end{romanenumerate}
\end{lemma}

\begin{proof}
Let us prove that \(\rrN\) is a right congruence.
Assume \(u \rrN v\), i.e., \(\post^{\cN}_{u}(I) = \post^{\cN}_{v}(I)\).
Given \(x \in \Sigma^*\), we have that,
\[\post^{\cN}_{ux}(I) = \post^{\cN}_{x}(\post^{\cN}_{u}(I)) = \post^{\cN}_{x}(\post^{\cN}_{v}(I)) = \post^{\cN}_{vx}(I) \enspace .\]
Therefore, \(ux \rrN vx\).

Now, let us prove that \(\rlN\) is a left congruence.
Assume \(u \rlN v\), i.e., \(\pre^{\cN}_{u}(F) = \pre^{\cN}_{v}(F)\).
Given \(x \in \Sigma^*\), we have that,
\[\pre^{\cN}_{xu}(F) = \pre^{\cN}_{u}(\pre^{\cN}_{x}(F)) =  \pre^{\cN}_{v}(\pre^{\cN}_{x}(F)) = \pre^{\cN}_{xv}(F) \enspace .\]
Therefore, \(xu \rrN xv\).

Finally, \(P_{\rrN}\), the finite partition induced by \(\rrN\).
We show that \(P_{\rrN}(\lang{\cN}) = \lang{\cN}\).
First note that \(L \subseteq P_{\rrN}(\lang{\cN})\) by the reflexivity of the equivalence relation \(\rrN\).
On the other hand, we prove that for every \(u \in \Sigma^*\), if \(u \in P_{\rrN}(\lang{\cN})\) then \(u \in \lang{\cN}\).
By hypothesis, there exists \(v \in \lang{\cN}\) such that \(u \rrN v\), i.e., \(\post^{\cN}_{u}(I) = \post^{\cN}_{v}(I)\).
Since \(v \in \lang{\cN}\) then \(\post^{\cN}_{v} \cap~ F \neq \emptyset\).
Therefore, \(\post^{\cN}_{u} \cap~ F \neq \emptyset\) and we conclude that \(u \in L\).

The proof of \(P _{\rlN}(L) = L\) goes similarly.
\end{proof}

\end{document}